\documentclass[11pt]{article}
\usepackage{fullpage}
\usepackage{authblk}
\usepackage[utf8]{inputenc}
\usepackage{amsmath,amssymb,bm,amsthm}
\usepackage{subfigure}
\usepackage{graphicx}
\usepackage{dcolumn}
\usepackage{bm}
\usepackage[mathlines]{lineno}
\usepackage{booktabs}
\usepackage{color}
\newcounter{one}
\usepackage{url}

\usepackage{textcase}

\numberwithin{equation}{section}

\usepackage{colortbl}
\usepackage{tabularx}
\usepackage{verbatim}
\usepackage{multirow}

\usepackage[T1]{fontenc} 
\usepackage{lmodern}
\usepackage{bbm}
\usepackage[utf8]{inputenc}
\usepackage{amsfonts}
\usepackage{array}

\usepackage{bbm}
\usepackage{enumitem}
\usepackage{umoline}
\usepackage[usenames,svgnames]{xcolor}
\usepackage[numbers]{natbib}
\usepackage[hyperindex,breaklinks]{hyperref}
\hypersetup{
     colorlinks=true,       		
     linkcolor=Navy,          	
     citecolor=Navy,            
     filecolor=Navy,      		
     urlcolor=Navy,           	
    runcolor=cyan,
 }

\usepackage{graphicx}
\usepackage{amsfonts}
\usepackage{amssymb}
\usepackage{amsmath}
\usepackage{bbm}
\usepackage{enumitem}
\usepackage{braket}

\usepackage{color}

\newcommand{\tr}[0]{ {\rm tr}}

\newcommand{\half}[1]{{ \rm h}}
\newcommand{\Oorderof}{\mathcal{O}}
\newcommand{\orderof}[1]{\Oorderof(#1)} 

\newcommand{\for}[0]{\quad \textrm{for} \quad}

\newcommand{\dist}{d}

\newcommand{\co}{{\rm c}}

\newcommand{\diam}{{\rm diam}}

\def\beq{\begin{equation}}
\def\eeq{\end{equation}}
\def\nbeq{\begin{equation*}}
\def\neeq{\end{equation*}}
\def\<{\langle}
\def\>{\rangle}

\def\tr{{\rm tr}}

\newcommand{\qmexp}[1]{\langle#1\rangle}

\newtheorem{theorem}{Theorem}[section]
\newtheorem{Theorem}[theorem]{Theorem}
\newtheorem{lemma}[theorem]{Lemma}

\newtheorem{assump}[theorem]{Assumption} 
\newtheorem{definition}[theorem]{Definition}  
\newtheorem{prop}[theorem]{Proposition} 

\theoremstyle{definition}
\newtheorem{remark}[theorem] {Remark}

\newcommand{\bal}[2]{#1[#2]}

\newcommand{\br}[1]{\left( #1 \right)}
\newcommand{\abs}[1]{\left | #1 \right|}
\newcommand{\brr}[1]{\left[ #1 \right]}
 \newcommand{\norm}[1]{\left \|  #1 \right \|}

\usepackage{mathtools}
\def\multiset#1#2{\ensuremath{\left(\kern-.3em\left(\genfrac{}{}{0pt}{}{#1}{#2}\right)\kern-.3em\right)}}

\newcommand{\brrr}[1]{\left\{ #1 \right\}}

\title{Symmetry-enhanced Lieb-Robinson bounds for a class of Bose-Hubbard type Hamiltonians}

\author[1,2]{Tomotaka Kuwahara\thanks{tomotaka.kuwahara@riken.jp}}
\author[3]{Marius Lemm\thanks{marius.lemm@uni-tuebingen.de}}
\affil[1]{Analytical quantum complexity RIKEN Hakubi Research Team, RIKEN Cluster for Pioneering Research (CPR)/ RIKEN Center for Quantum Computing (RQC), Wako, Saitama 351-0198, Japan}
\affil[2]{PRESTO, Japan Science and Technology (JST), Kawaguchi, Saitama 332-0012, Japan}
\affil[3]{Department of Mathematics, University of T\"ubingen, 72076 T\"ubingen, Germany}

\date{July 8, 2025}

\begin{document}
\maketitle

\begin{abstract}
Several recent works have derived Lieb-Robinson bounds (LRBs) for Bose-Hubbard-type Hamiltonians.  For certain structured initial states, e.g., vacuum perturbations or near-stationary states,  information propagates with velocity $v \leq C$ . However, for general bounded-density initial states, it was shown by the first author, Vu, and Saito that the velocity can grow in time as $v \sim t^{D-1}$, where $D$ is the spatial dimension --- demonstrating the possibility of accelerated information spreading in bosonic systems.
In this work, we introduce a new perspective on this phenomenon: we show that translation invariance combined with local $p$-body repulsion ($n^p$ with $p > D+1$) qualitatively alters the propagation behavior, leading to a bound of the form $v \sim t^{\frac{D}{p - D - 1}}$ for general bounded-energy-density initial states. In particular, this establishes for an almost-linear light cone at large $p$, in stark contrast to the previously found accelerated regimes.
Our result identifies symmetry-driven constraints as a new mechanism for suppressing propagation speed in bosonic systems and thereby reframes the scope of what types of LRBs can hold. We further provide matching examples showing that, under the given assumptions, this bound is sharp --- no further improvement in the power of $t$ is possible without invoking additional dynamical constraints.

\end{abstract}


%
%
%
%
%
%
%
%
%



\section{Introduction}
Lieb-Robinson bounds (LRBs) \cite{lieb1972finite} provide a speed limit on the propagation of quantum information in a quantum many-body lattice system with local interactions. The standard LRBs, see \cite{lieb1972finite,hastings2004decay,nachtergaele2006lieb}, concern two local bounded operators ( ``observables'') $O,\tilde O$ and are of the form 
\begin{align}
\label{eq:locality_commutator}
\|[O(t),\tilde O]\|\leq C\|O\| \|\tilde O\| e^{C(vt-d(O,\tilde O))},
\end{align}
where $O(t)=e^{\mathrm{i}tH}O e^{-\mathrm{i}tH}$ is the Heisenberg time evolution generated by the Hamiltonian $H$. Here, $d(O,\tilde O)=\mathrm{dist}(\mathrm{supp\, O},\mathrm{supp}\, \tilde O)$ and the support of an observable $O$ is defined as the smallest lattice region $X$ such that $O$ acts as the identity on $X^c$, i.e., $O=O_X\otimes \mathrm{Id}_{X^c}$. 

The key quantity that characterizes the quantum information propagation is $v\geq 0$. This bound is called the Lieb-Robinson velocity. In the standard LRB, which holds for many-body systems with \textit{bounded} interactions, $v$ is proportional to the operator norm of the local interaction and independent of the time $t$. Consequently, the standard LRB \cite{lieb1972finite,hastings2004decay,nachtergaele2006lieb} is only meaningful for systems with bounded interactions such as quantum spin systems (and also lattice fermions \cite{nachtergaele2018lieb}). 

LRBs have grown into a rich subject that connects mathematical physics with  condensed matter physics, quantum information science, and high-energy physics. This success story started in the early 2000s \cite{hastings2004decay,hastings2005quasiadiabatic,bravyi2006lieb,nachtergaele2006lieb,nachtergaele2006propagation,hastings2007area,hastings2007quantum,nachtergaele2007multi,bachmann2012automorphic,kliesch2014locality}, kicked-off by the highly influential realization of Hastings that LRBs, which are dynamical bounds, are surprisingly effective tools to study equilibrium properties, e.g., clustering of correlations \cite{hastings2004decay,nachtergaele2006lieb}, the area law for entanglement entropy \cite{hastings2007area} and topological quantum phases \cite{hastings2005quasiadiabatic,bachmann2012automorphic}.
Since then, it has been realized that LRBs are among the few robust, effective, and versatile tools for analyzing quantum many-body systems in and out of equilibrium. More recently, they have also been used to limit information scrambling in the theory of quantum many-body chaos which is relevant for black hole physics and quantum information science \cite{lashkari2013towards,roberts2016lieb,chen2019finite,kuwahara2021absence}. For further background on LRBs, we refer to the reviews \cite{nachtergaele2010lieb,gogolin2016equilibration,chen2023speed}.

\subsection{Bosonic propagation bounds}
Since the standard LRBs require bounded interactions, they do not give meaningful bounds for lattice bosons or systems in continuous space. 

Special bosonic systems such as perturbations of harmonic oscillator systems have been considered \cite{cramer2008locality,nachtergaele2009lieb,woods2015simulating,woods2016dynamical}.
However, the results covering the paradigmatic \textit{Bose-Hubbard Hamiltonian} \cite{schuch2011information,wang2020tightening} were quite restricted, and proving robust LRBs for this Hamiltonian, in particular for a broad class of initial states, was a major methodological barrier. In the last three years, this barrier was broken and bosonic LRBs for Bose-Hubbard Hamiltonians have been derived by three different research groups by three different methods \cite{kuwahara2021lieb,faupin2022lieb,yin2022finite,kuwahara2022optimal,sigal2022propagation,lemm2023information}, which we summarize below. Similar bounds for systems in continuous space have also  recently been studied by related methods \cite{gebert2020lieb,arbunich2021maximal,breteaux2022maximal,breteaux2024light,hinrichs2024lieb}.


Part of the recent interest from the physics community in LRBs for Bose-Hubbard type Hamiltonians stems from the fact that the latter can be experimentally realized, finely tuned, and measured with high fidelity, in tabletop experiments with ultra-cold quantum gases in optical lattices \cite{denschlag2002bose,bloch2008many,greiner2008optical}. These lattice bosons provide some of the early examples of quantum simulators \cite{jaksch2004optical,bloch2012quantum,lewenstein2012ultracold,gross2017quantum,yang2020observation,kaufman2021quantum,su2023observation}. In this context, bosonic LRBs have been experimentally observed since about 10 years ago \cite{barmettler2012propagation,cheneau2012light,cheneau2022experimental} and their validity can provide important performance bounds on the use of quantum simulation and quantum information protocols \cite{bravyi2006lieb,kliesch2014lieb,epstein2017quantum,faupin2022lieb,kuwahara2022optimal,sigal2022propagation,lemm2023information}.

We briefly summarize the state-of-the-art bosonic LRBs for Bose-Hubbard type Hamiltonians on a finite graph $(\Lambda,\mathcal E)$. In this introduction, we consider Hamiltonians of the form 
\begin{equation}
\label{eq:HBHintro}
\begin{aligned}
&H= \sum_{\substack{i,j\in \Lambda:\\ i\sim j}} J_{i,j}(b_i^\dagger b_j+b_j^\dagger b_i) +\sum_{i\in\Lambda} w(n_i).
\end{aligned}
\end{equation}
where we write $i\sim j$ to express that $(i,j)\in\mathcal E$, i.e., that $i$ and $j$ are nearest neighbors.
 The first term represents boson hopping.  $J_{i,j}$ are the entries of a real symmetric matrix of size $|\Lambda|\times |\Lambda|$ and we focus on nearest-neighbor hopping. The second term in \eqref{eq:HBHintro} represents boson-boson density-dependent interaction. For the interaction $w:\mathbb N\to \mathbb R$ we use the convention that $\mathbb N=\{0,1,2,\ldots\}$. A prototypical case is that of constant nearest-neighbor hopping $J_{i,j}=\bar J$ for all  and quadratic local repulsion $w(n_i)=n_i(n_i-1)-\mu n_i$.

Since lattice bosons involve unbounded interactions, it is essential that any bosonic LRB should not involve operator norms (which are, after all, worst-case quantities) and instead control quantum information propagation for certain classes of ``well-behaved'' initial states $\rho$. In this vein, the following types of initial states have been considered with qualitatively different LRBs. 
The works \cite{schuch2011information,kuwahara2021lieb,yin2022finite,faupin2022lieb,kuwahara2022optimal,lemm2023information} proved a bosonic LRB with $v$ independent of time for various types of special initial states. More precisely, \cite{schuch2011information} covered states with all particles in a fixed finite region, \cite{kuwahara2021lieb} covered perturbations of stationary states, \cite{yin2022finite} covered the initial state $\rho=e^{-\mu N}$,  and \cite{faupin2022lieb,sigal2022propagation,lemm2023information} covered states with general particle-free subregions and long-range interactions.
Most relevant for our purposes is the work \cite{kuwahara2022optimal} which proved a bosonic LRB with $v\sim t^{D-1}$ for initial states with suitably bounded particle density.
(As one would expect, all these bounds have in common that the LR velocity $v$ is independent of system size, which is necessary to have a meaningful bound in the thermodynamic limit. Even this arguably basic requirement is far from trivial for bosons \cite{wang2020tightening}.)

Any LRB is merely an upper bound on information transport. It is always of interest to complement it by a qualitatively similar lower bound, i.e., to specify situations when the transfer is indeed of a prescribed speed. This question is particularly pressing in the situation considered in \cite{kuwahara2022optimal} because there, for $D\geq 2$, the velocity bound $v\sim t^{D-1}$ grows with time. This bound in principle allows for \textit{information acceleration}, by which we mean that for any valid LRB of the form
$\abs{\mathrm{tr} (\rho([O(t),\tilde O]))}\lesssim f(d(O,\tilde O)-vt)$,
with $f$ decaying at large arguments, the effective velocity $v$ must grow with time.

It is unclear whether this type of acceleration truly occurs for Bose-Hubbard Hamiltonians starting from bounded-density initial states. One reason for this is that, mathematically, the role of symmetries such as translation invariance in LRBs remains poorly understood. While symmetries often simplify dynamics, it is highly nontrivial to rigorously incorporate their effects into bounds on information propagation. In fact, it is counterintuitive that translation invariance alone is insufficient to suppress bosonic accumulation or slow down the spread of information.

In view of this, we put forward the following open, fundamental, and pressing question:

{~}

\textit{Question 1:} Consider a time-independent Bose-Hubbard type Hamiltonian such as \eqref{eq:HBHintro}, prepared in a bounded-density initial state $\rho$. Does it display a finite speed of information propagation?

{~}

\noindent
This question remains unresolved even in the case of the standard Bose-Hubbard model with nearest-neighbor hopping and on-site quadratic interaction, underscoring its conceptual importance.

Naive intuition would suggest that information acceleration cannot occur in a closed quantum system governed by unitary dynamics that is not subjected to any external forcing. However, this naive intuition is provably incorrect. Indeed, there exist \textit{explicit examples of simple, time-independent, and translation-invariant bosonic Hamiltonians which rigorously display information acceleration}, meaning that no LRB can hold with $v$ independent of time. Indeed, Eisert and Gross \cite{eisert2009supersonic} exhibited a one-dimensional, translation-invariant, time-independent lattice boson Hamiltonian where any LRB must have $v\sim e^{Ct}$, i.e., information accelerates exponentially. While this Hamiltonian is \textit{not} of Bose-Hubbard type \eqref{eq:HBHintro}, it impressively shows that information acceleration for bosons is a realistic possibility that has to be confronted in a model-dependent way, unlike for quantum spin systems. In the words of \cite{eisert2009supersonic}: ``Since we cannot
rule out the presence of accelerating excitations by imposing natural assumptions, one must not take the existence of
a finite speed of sound in any bosonic model for granted.''

\textit{Moreover, the work \cite{kuwahara2022optimal} also exhibits a protocol for a time-dependent Bose-Hubbard type Hamiltonian which rigorously propagates quantum information with $v\sim t^{D-1}$.} This has also sparked heuristic arguments that the scaling $v\sim t^{D-1}$ could be optimal for time-\textit{in}dependent Bose-Hubbard type Hamiltonians as well.
Hence, a different, weaker variant of Question 1 is to ask to what extent the power of $t$ in the LRB for bounded initial states in \cite{kuwahara2022optimal} is in fact optimal.\\

\textit{Question 1':} Consider a time-independent Bose-Hubbard type Hamiltonian such as \eqref{eq:HBHintro}, prepared in a bounded energy-density initial state $\rho$. Can the LRB from \cite{kuwahara2022optimal}, which gives $v\sim t^{D-1}$, be qualitatively enhanced to $v\sim t^\alpha$ with $\alpha<D-1$ under physically meaningful assumptions?\\  

\noindent
Addressing this question provides the first rigorous indication that translation invariance, even in the absence of disorder or localization, can qualitatively affect the scaling behavior of light cones in interacting bosonic systems. Any partial solution to this question would contribute a first conceptual step toward understanding how symmetry constraints can qualitatively influence the Lieb-Robinson light cone.

In this paper, we provide a positive answer to Question 1'. We work under two additional physical assumptions: (i) strong on-site repulsion $n_i^p$ with $p>D+1$ and (ii) translation-invariance; see the next paragraph. In words, we prove that the velocity bound $\sim t^{D-1}$ proved for a broad class of initial states in \cite{kuwahara2022optimal} improves to  $t^{\epsilon}$ for sufficiently large $p$-body repulsion. So a particular consequence is the first proof of an almost-linear light cone for a class of Bose-Hubbard Hamiltonians starting from the most physically relevant class of positive energy density initial states. In particular, this includes the most relevant class of  Mott states. Our enhanced LRB is obtained by leveraging for the first time energetic dynamical constraints in the proof strategy developed in  \cite{kuwahara2022optimal}. We expect this energetic perspective to also be useful in other settings.
It is worth emphasizing that prior to this work, it was not even physically clear whether one could expect an almost linear light cone  for general bounded density initial states.  We view our result as another puzzle piece in the quest to understand the surprising richness of bosonic transport behavior, a richness that was to our knowledge first emphasized by Eisert and Gross in \cite{eisert2009supersonic}.\\

We remark that translation invariance of the system Hamiltonian is a particularly natural condition of interest. It is physically meaningful and it intuitively limits the possibility of the boson concentration on some particular regions, which is a central puzzle piece in the explicit protocol displaying the accelerating information transfer. In fact, naive intuition might suggest that the translation invariance immediately leads to bounded local particle numbers with high probability and thus a finite Lieb-Robinson velocity: since a translation-invariant state with large boson occupation at a site must also have an equal chance of having large boson occupation at all other sites, one could hope that this should lead to a combinatorial suppression of the possibility of such problematic states. 
Unfortunately, this naive combinatorial picture is incorrect, because quantum mechanics allows for the phenomenon of ``superpositions of `bad' states''. This observation is already implicit in \cite{kuwahara2022optimal} and it underlies many of the technical difficulties. To show that this is a real possibility, in Appendix \ref{app:bad}, we construct an explicit family of translation-invariant states which have significant boson concentration, while also satisfying the energetic bounds we use here. These examples show that there are limitations to use only certain dynamical constraints (including translation-invariance).

Our work initiates a new direction in the analysis of bosonic Lieb-Robinson bounds. We identify physically natural and experimentally relevant conditions under which the known scaling $v \sim t^{D-1}$, though widely regarded as optimal, can be qualitatively improved.
This leads to a new dynamical structure, which we can refer to as a symmetry-enhanced Lieb-Robinson light cone. Here, the suppression of information propagation arises not from disorder or localization, but from global symmetry and energetic constraints.
To our knowledge, this is the first work to demonstrate that global symmetry alone can lead to a qualitative modification of the light cone structure in many-body systems.

\subsection{Statement of result}
We now state Theorem \ref{thm:maininformal}, which is a consequence of the more general main results, Theorems \ref{thm:lrbtr} and \ref{thm:lrbav}.
It establishes two enhanced LRBs which are qualitatively different depending on how the size of the commutator $[O(t),\tilde O]$ is measured when tested with respect to the initial state $\rho$; either in trace norm, $\|A\|_1=\mathrm{tr} |A|$, or in expectation. We write $\mathcal F(\ell^2(\Lambda))$ for the Fock space over a finite graph $\Lambda$, keeping in mind the common convention that we denote both the graph and its vertex set by $\Lambda$. We write $\mathcal S(\mathcal H)$ for the set of quantum states on a Hilbert space $\mathcal H$, i.e., $\rho\in\mathcal S(\mathcal H)$ iff $\rho$ is a positive semi-definite trace class operator of trace equal to $1$. 
We denote the local particle number on a set $X$ by $n_{X}=\sum_{i\in X} n_i$. 

We write $i\sim j$ to express that two vertices $i$ and $j$ of a graph are nearest neighbors, i.e., that they are connected by an edge. 
If the graph is equipped with a notion of translation $\mathcal T$ (such as the discrete $D$-dimensional torus below), then we can lift this to the bosonic Fock space by $\Gamma(\mathcal T)$. We recall that $\Gamma(\mathcal{T})$ means acting with the translation operator on all particles; see \cite[Ch. X.7]{reed1975ii} for the definition and basic properties of $\Gamma(\cdot)$. We call operators that commute with $\Gamma(\mathcal T)$ translation-invariant.

Below, we introduce suitable time-dependent velocities $v_1(t)$ and $v_{\mathrm{ex}}(t)$ in order to compare with other results.

\begin{theorem}[Main example of results]
\label{thm:maininformal}
    We consider a nearest-neighbor, translation invariant Bose-Hubbard Hamiltonian \eqref{eq:HBHintro} on a discrete $D$-dimensional torus of side length $L\geq 2$. That is, we set
    \[
    \Lambda_L=(\mathbb Z/L\mathbb Z)^D
\]
with the edge set
\[
\mathcal E_L:=\{e=(i,j)\in \Lambda_L\,:\, \sum_{j=1}^D \| (i-j) \,\mathrm{mod}\, L \|_{\ell^1}=1\}
\]
where $\mathrm{mod}\, L$ is taken componentwise in $\mathbb Z^D$.

        Let $p\in(1,\infty)$ and let $\tilde w:\mathbb N\to \mathbb R$ be a function for which there exists $c_{\tilde w},\epsilon>0$ such that $| w(n)|\leq  c_{\tilde w} (n+1)^{p-\epsilon}$ for all $n\geq 0$. We consider a  local interaction of the form
     \begin{equation}
    W=\sum_{i\in\Lambda} \brr{n_i^p+\tilde w(n_i)}.
   \end{equation}
   
   Let  $\rho\in \mathcal S(\mathcal F(\ell^2(\Lambda_L)))$ be a translation-invariant quantum state that has bounded expected energy density, i.e., 
   \begin{equation}
         \frac{\tr\brr{\rho H}}{|\Lambda_L|}=E_\rho<\infty.
   \end{equation}

   Let $X_0\subset \Lambda$. There exists a constant $C=C(D,J,c_V,\epsilon,E_\rho,X_0)$ such that the following holds for every pair of bounded observables $O,\tilde O$ on $\mathcal F(\ell^2(\Lambda_L)$ such that $\mathrm{supp}\, O=X_0$ and $[O, n_{X_0}]=0$.

    \begin{enumerate}[label=(\Roman*)]
    
    \item Suppose that $p>2D+2$ and set $v_1(t)=t^{\frac{D}{p/2-D-1}}$. Then
   \begin{equation}
   \left\|\rho([O(t),\tilde O])\right\|_1\leq C\|O\|\|\tilde O\| \left(\frac{v_1(t) t}{d(O,\tilde O)}\right)^{p/2-D-1},\qquad \textnormal{for } t,d(O,\tilde O)\geq 1.
   \end{equation}
    \item Suppose that $p>D+1$ and set $v_{\mathrm{ex}}(t)=t^{\frac{D}{p-D-1}}$. Then
   \begin{equation}
    \abs{\mathrm{tr} (\rho([O(t),\tilde O]))}\leq C \|O\|\|\tilde O\| \left(\frac{v_{\mathrm{ex}}(t) t}{d(O,\tilde O))}\right)^{p-D-1},\qquad \textnormal{for } t,d(O,\tilde O)\geq 1.
   \end{equation}
\end{enumerate}

\end{theorem}

Theorem \ref{thm:maininformal} stems from the more general Theorems \ref{thm:lrbtr} and \ref{thm:lrbav}. These theorems offer results that are applicable to a wider variety of Bose-Hubbard Hamiltonians~\eqref{eq:HBHintro} (which do not necessarily have to be translation-invariant) as long as they involve finite-range interactions and adhere to a specific condition over time:  $\mathrm{tr}\brr{\rho_0(t)n_i^p}\leq C_p$ for all $i\in\Lambda$. 
We then demonstrate that this condition is satisfied when two physical criteria are met: (i) translation-invariance and (ii) strong on-site repulsion, which then imply Theorem \ref{thm:maininformal} as a consequence. As Theorems \ref{thm:lrbtr} and \ref{thm:lrbav} show, the assumption in Theorem \ref{thm:maininformal} that the initial observable $O$ is particle-number preserving can be substantially relaxed --- it suffices to assume control on how many particles can be created by $O$, cf.\ Definition \ref{defn:q0}.

We compare and contrast our results to those of \cite{kuwahara2022optimal}.

    \begin{itemize}
    \item As in \cite{kuwahara2022optimal}, we are able to cover a very broad class of initial states, namely translation-invariant ones of finite energy density. This includes the experimentally most relevant class of Mott states \cite{bloch2008many,cheneau2012light,cheneau2022experimental}. 
        \item   Theorem \ref{thm:maininformal} establishes an LRB which bounds the velocity of information propagation by $v_{1}(t)=t^{\frac{D}{D+1-p/2}}$ respectively $v_{\mathrm{ex}}(t)= t^{\frac{D}{D+1-p}} $ depending on how the size of the commutator $[O_1(t),O_2]$ is measured (either in trace norm or in quantum expectation). Note that $\lim_{p\to\infty}v_{1}(t)=\lim_{p\to\infty}v_{\mathrm{ex}}(t)=1$, so for sufficiently large $p$, we obtain an approximately constant maximal velocity bound, closer to the case of quantum spin systems. This is a manifestation of the fact that the system gets increasingly rigid as $p$ grows. This is the first proof of an almost-linear light cone for a class of Bose-Hubbard Hamiltonians for a broad class of initial states.

        \item The smallest $p$-values such that the velocities $v_{1}(t)$ and $v_{\mathrm{ex}}(t)$ scale more favorably than the bound $v\sim t^{D-1}$ from \cite{kuwahara2022optimal} can be formulated as the conditions $p>2D+2+\tfrac{2D}{D-1}$ (for the trace norm) or $p>D+1+\tfrac{D}{D-1}$ (for the expectation). For the physical dimensions $D=2$ and $D=3$, the smallest integer value of $p$ for which our result yields an improved scaling compared to \cite{kuwahara2022optimal} is $p= 6$. That is, in physical terms, we prove that six-body boson-boson repulsion yields qualitatively slower information transport than what was expected based on \cite{kuwahara2022optimal}. This also shows that the ``accumulation scenario'' in Figure 4a of \cite{kuwahara2022optimal}, which, as we recall, produces information propagation with $v\sim t^{D-1}$, is provably excluded in our setting. (We emphasize that our result is not in any contradiction to the rigorous protocol in \cite{kuwahara2022optimal}. The reason is that the protocol requires a special time-dependent Bose-Hubbard Hamiltonian that neither enjoys strong repulsion nor translation-invariance.) Our result highlights how energetic constraints play an important role in controlling the dynamical spreading of quantum information in bosonic systems.
        
\item Regarding methods, our proof heavily draws on techniques developed in the 100+ page paper \cite{kuwahara2022optimal}. However, we add a crucial new twist by exploiting energetic constraints, whereas \cite{kuwahara2022optimal} exclusively used particle propagation bounds without making use of the interaction energy at all. We believe that the new perspective of using energetic dynamical bounds can be of use in similar quantum dynamics problems in the future.




     \end{itemize}

Under the assumptions of Theorem \ref{thm:maininformal}, we also obtain an improvement of the particle propagation bound for arbitrary moments (Main result 1 in \cite{kuwahara2022optimal}) by an interpolation argument; see Proposition \ref{prop:interpolation}.

\subsection{Discussion}

We comment on the role of assumptions of (i) strong repulsion and (ii) translation invariance and the connection of our work to \cite{kuwahara2022optimal} at a conceptual and technical level.

Local interactions of the form $n_i^p$ with an integer exponent $p>2$ are not part of the standard Bose-Hubbard Hamiltonian, but they have been discussed in the physics literature to arise from higher-order local repulsion between multiple particles \cite{will2010time, mark2011precision,petrov2014elastic,petrov2014three,mondal2020two}. The most significant case for physical applications is when $p=2$, representing two-body interactions in the typical Bose-Hubbard Hamiltonian. Our findings do not extend to this case. On the one hand, to further improve our result to include $p=2$, we anticipate the need for a more intricate analysis of how energy flows locally and how higher-order particle correlations develop dynamically. Dealing with these aspects of local dynamics will necessitate the introduction of new insights and ideas. On the other hand, it is by now well-established that bosonic information propagation is subtle and can produce rapidly accelerating information propagation in some cases \cite{eisert2009supersonic,kuwahara2022optimal}, it is not clear at all, even at a physical level, if the LRB from \cite{kuwahara2022optimal}  can be improved for the original Bose-Hubbard model with $p=2$ in the most interesting regime of $D\leq 3$ for a broad class of initial states. 

Second, concerning \textit{translation-invariance}, we remark that naive combinatorial heuristics suggest that translation-invariance by itself could be enough to obtain an LRB with $v\sim 1$. These heuristics go as follows: Note first that large speed $v$ comes from local boson accumulation. Now consider a time-evolved state that is ``bad'' in the sense that it exhibits extreme boson accumulation at a given site. By translation-invariance, all of its translates are bad states in the same sense.  Therefore, one could hope that when looking at a fixed site the total state has a small chance of having boson accumulation nearby, which should in turn yield $v\sim 1$. Unfortunately, this combinatorial heuristic is false.
This can be seen by constructing translation-invariant states by superposing states with periodic boson accumulation on highly occupied lines. This shows that, \textit{unfortunately, translation-invariance is not as useful as one might naively hope based on classical probabilistic intuition}. We give examples of such states in Appendix \ref{app:bad}. This is an instance of the more general phenomenon that ``superpositions of bad states can behave well''. This genuinely quantum phenomenon is at the heart of the challenge of using combinatorial-probabilistic arguments in the quantum many-body context. In fact, excluding a similar ``superpositions of bad states,'' kind of scenario also constitutes the main technical work in \cite{kuwahara2022optimal}.
There, the main dynamical constraint used is that the speed of particle transport is uniformly bounded in time (see also \cite{faupin2022maximal,van2023optimal,lemm2023information,lemm2023microscopic} for more general bounds of this type, e.g., for long-range hopping). More precisely, one shows that starting from a bounded-density initial state, within time $t$ the number of particles on a $1D$ subregion is at most $t^{D-1}$ and this is ultimately the reason for $v\sim t^{D-1}$. By considering this scenario, one sees that particle propagation bounds alone are insufficient to improve over $v\sim t^{D-1}$. Our result shows that translation-invariance and global energy conservation, for a sufficiently strong repulsion, are sufficient for an improvement. On a technical level, we replace the exponential but time-dependent bound on the local particle number in \cite{kuwahara2022optimal} with a polynomial time-independent bound that we derive from translation-invariance and energy conservation. Then we draw on the idea of concatenating and short-time LRBs from \cite{kuwahara2021lieb} a and certain estimates on comparison with truncated dynamics from \cite{kuwahara2022optimal}. For this, we have to take extra care to handle the slower decay which is polynomial in our case. In addition, we observe that the bound on the expectation value can be improved compared by re-ordering certain expansion terms in the Duhamel expansion of the truncated dynamics.


\subsection{Organization of the paper}
 In Section \ref{sect:setup}, we introduce the formal setup and state the main results, Theorems \ref{thm:lrbtr} and \ref{thm:lrbav} which imply Theorem \ref{thm:maininformal} as an easy corollary. We also state an improved bound on the particle propagation, Proposition \ref{prop:interpolation}, that may be of independent interest.

 Section \ref{sect:strategy} contains the overall proof strategy for Theorems \ref{thm:lrbtr} and \ref{thm:lrbav}. It shows how these theorems are reduced to suitable short-time LRBs (Theorems \ref{Theorem_for_small_time_evo} and \ref{Theorem_for_small_time_evo_average}) by a suitable concatenation procedure. Theorem \ref{Theorem_for_small_time_evo} is proved in Section \ref{sect:st1} and Theorem \ref{Theorem_for_small_time_evo_average} is proved in Section \ref{sect:st2}. We summarize our findings and give an outlook in Section \ref{sect:conclusion}.

 In Appendix \ref{app:bad}, we construct translation-invariant quantum states with bounded energy density and significant boson concentration. These examples serve to show that the propagation bounds we prove are sharp conditional on the kinds of dynamical constraints that we use (translation-invariance and conservation of the first, even of the second, moment of the Hamiltonian). Further improvements will therefore need to leverage additional dynamical constraints and likely a much finer understanding of the energy flow on microscopic length scales that is generated by the many-body dynamics.

 We heavily rely on several technical results from \cite{kuwahara2021lieb,kuwahara2022optimal} as important ingredients to prove our enhanced LRB. This is in some sense natural given that we seek to enhance a bound that took considerable effort to prove. Whenever we use results from these works, we include their precise statement, but we do not repeat the proofs in order to keep the paper at a moderate length.
 
 \section{Setup and main results}\label{sect:setup}
\subsection{Graph-theoretic notation and setup}
Let $(\Lambda,\mathcal E)$ be a finite graph. Recall that we write $i\sim j$ for $(i,j)\in\mathcal E$, i.e., if $i$ and $j$ are nearest neighbors. As is common in the literature on LRBs, we occasionally refer to vertices as ``sites''.
For an arbitrary subset $X\subseteq \Lambda$, we denote its cardinality by $|X|$.

For arbitrary subsets $X, Y \subseteq \Lambda$, we define $\dist_{X,Y}$ to be the graph distance between $X$ and $Y$ (i.e., the length of the shortest path between $X$ and $Y$). If $X\cap Y \neq \emptyset$, we set $\dist_{X,Y}=0$. 
When $X$ is composed of only one element (i.e., $X=\{i\}$), we denote $\dist_{\{i\},Y}$ by $\dist_{i,Y}$ for simplicity. Similarly, when $X=\{i\}$ and $Y=\{j\}$, we write $\dist_{i,j}$ for $\dist_{X,Y}$.

We denote the complementary subset of $X$ by $X^\co := \Lambda\setminus X$ and we define its associated boundary (or surface) subset by 
\[
\partial X:=\{ i\in X| \dist_{i,X^\co}=1\}.
\]
We also define the diameter $\diam(X)$ in a slightly non-standard, but convenient way as follows,
\begin{align}
\diam(X):  =1+ \max_{i,j\in X} (\dist_{i,j}).
\end{align}
We will track locality mainly in terms of extended subsets, which we denote by
\begin{equation}\label{def:bal_X_r}
\bal{X}{r}:= \{i\in \Lambda| \dist_{X,i} \le r \},\qquad r>0.
\end{equation}
We use the convention that $\bal{X}{0}=X$. 

We also write $\leq$ for operator inequalities, i.e., for two bounded Hermitian operators $A\leq B$ means that $B-A$ is a positive semidefinite operator.

\begin{assump}[$D$-dimensional graph]
\label{ass:graph}
We assume that the graph $\Lambda$ is $D$-dimensional, in the sense that its surface growth is controlled via a lattice parameter $\gamma \ge 1$ by\footnote{It is trivial that such a constant $\gamma<\infty$ exists for any finite graph for any $D$, but the point is that the bounds we prove will only depend on the graph $\Lambda$ through the constant $\gamma$. This means that the bounds are stable in the limit where $\Lambda$ is large and $\gamma$ stays fixed (which is the case in most practical situations, e.g., for large boxes or tori in $\mathbb Z^D$).}
\begin{align}
\label{supp_parameter_gamma_X_i}
 \abs{ \partial (i[\ell]) } \le \gamma \ell^{D-1} ,\qquad \textnormal{for all }  \ell\geq0.
\end{align}
\end{assump}
We remark that the surface growth bound \eqref{supp_parameter_gamma_X_i} implies the volume growth bound $\abs{ i[\ell] } \le (\gamma+1) \ell^{D}.$ The reverse implication does not hold in general; see, e.g., \cite[Section 2.3]{tessera2007volume}.

\subsection{Assumptions on the Hamiltonian and observables}
We consider the bosonic Fock space $\mathcal F(\ell^2(\Lambda))$ with bosonic creation and annihilation operators $\{b_i^\dagger,b_i\}_{i\in\Lambda}$. Each $b_i$ annihilates the vacuum vector and the collection $\{b_i^\dagger,b_i\}_{i\in\Lambda}$ satisfies the canonical commutation relations. We write $n_i=b_i^\dagger b_i$ for the particle number operator at site $i\in\Lambda$ and  $\mathcal N=\sum_{i\in\Lambda}n_i$ for the total particle number operator.

Let $p\in(1,\infty)$. On $\mathcal F(\ell^2(\Lambda))$, we consider general Bose-Hubbard-type Hamiltonians of the form 
\begin{align}
&H= H_0+ W, 
\notag \\
\textnormal{where } &H_0\coloneqq  \sum_{\substack{i,j \in \Lambda:\\ i\sim j}}J_{i,j} (b_i b_j^\dagger +b_i^\dagger b_j ), \label{def:Ham}\\
\textnormal{and } &\bar{J}:=\sup_{i\sim j} |J_{i,j}|<\infty.
\end{align}
We recall that $i\sim j$ means that $i$ and $j$ are nearest neighbors. 
The interaction $W$ satisfies the following assumption.

\begin{assump}[Finite-range interaction]\label{ass:V}
 We assume that the interaction $W$ is of the form    \begin{equation}
 \label{eq:Wdef}
W=\sum_{\substack{X\subset \Lambda:\\
\mathrm{diam}(X)\leq k}} w_X(\{n_i\}_{i\in X})
   \end{equation}
   for some constant $k\geq 1$. That is, $W$ 
   depends only on the occupation numbers $n_i$ and is of finite range $k$.
\end{assump}
Notice that $H$ commutes with the total particle number operator $\mathcal N$ and can therefore be block-diagonalized as $H=\bigotimes_{N\geq 0} H_N$ where each $H_N$ is a bounded self-adjoint operator. Therefore, $H$ is self-adjoint on the domain
\[
\mathcal D(H)=\left\{
\varphi=(\varphi)_{N\geq0}\in \mathcal F(\ell^2(\Lambda))\,:\,
\sum_{N\geq 0} \|H_N\varphi_N\|^2<\infty\right\},
\]
see, e.g., \cite[Appendix A]{faupin2022lieb} for a proof.
By Stone's theorem, $e^{-\mathrm{i}tH}$ is a strongly continuous one-parameter group of unitaries and we can thus define the Heisenberg time evolution of any bounded observable $O$ by $O(t)=e^{\mathrm{i}tH}Oe^{-\mathrm{i}tH}$. Any restriction of the Hamiltonian to subgraphs of $\Lambda$ is also self-adjoint for the same reasons.


Our main results apply to operators $O$ that are block-diagonal in the local particle number basis. The size of the block is captured by the parameter $q_0$ which physically represents the maximal change of local particle number that can be achieved by the operator $O$ under consideration. Here and in the following, given a self-adjoint operator $A$ and $\alpha\in \mathbb R$, we write $\Pi_{A>\alpha}$ for the spectral projector onto the spectral subspace where $A>\alpha$, and we define $\Pi_{A<\alpha}$, $\Pi_{A=\alpha}$, etc., analogously.
 
\begin{definition}[Maximal change of local particle number]\label{defn:q0}
    Let $X_0\subset\Lambda$ be bounded. Let $O$ be an observable supported on $X_0$. We say that $O$ changes the local particle number by at most $q_0\geq 0$, if
    \begin{align}
\label{O_X_condition/_norm_lemma_operator}
\Pi_{n_{X_0} > m+q_0} O_{X_0} \Pi_{n_{X_0}=m}=0, \quad\textnormal{ and }\quad \Pi_{n_{X_0} < m-q_0} O_{X_0} \Pi_{n_{X_0}=m} =0,\qquad  \forall m>0.
\end{align}
\end{definition}
Of course, any particle-number conserving observable on $X_0$ satisfies Definition \ref{defn:q0} with $q_0=0$. We shall derive the LRB for observables $O$ that satisfy this definition for some $q_0<\infty$. While we assume $q_0=0$ in Theorem \ref{thm:maininformal} for simplicity of presentation, the result immediately generalizes to any fixed $q_0$, because Theorems \ref{thm:lrbtr} and \ref{thm:lrbav} hold for any fixed $q_0$.

\subsection{Statement of main results}
Our main results (Theorems \ref{thm:lrbtr} and \ref{thm:lrbav}) are LRBs of the form (I) and (II)  appearing in Theorem \ref{thm:maininformal} for much more general Bose-Hubbard type Hamiltonians \eqref{def:Ham} and initial states, but assuming the moment bound condition \eqref{basic_assump_moment_func}. 
 Our main example where the moment bound \eqref{basic_assump_moment_func} can be verified is in the setting of Theorem \ref{thm:maininformal}, i.e., in the presence of (i) strong local repulsion and (ii) translation invariance. Indeed, we confirm in Section \ref{ssect:verify} below that Theorem \ref{thm:maininformal} follows from Theorems \ref{thm:lrbtr} and \ref{thm:lrbav} in a straightforward way.
 
We write $\rho$ for the initial quantum state, which is a positive semidefinite operator on $\mathcal F(\ell^2(\Lambda))$ with $\mathrm{tr}\brr{\rho}=1$. Recall that $v_1(t)=t^{\frac{D-1}{p/2-D-1}}$, 
as defined in Theorem~\ref{thm:maininformal}.

\begin{theorem}[Enhanced LRB for trace norms]
\label{thm:lrbtr}
Let $H$ be of the form \eqref{def:Ham} with $W$ satisfying Assumption \ref{ass:V}.
Assume that there exist $T\in [0,\infty]$ and $p>2D+2$ such that
\begin{align}
\sup_{i\in\Lambda}\tr\brr{ \rho n_i^p(t) } \le C_p,\qquad \textnormal{for all }t\in [0,T].
\label{basic_assump_moment_func}
\end{align}
Fix $X_0\subset \Lambda$ and $q_0\geq 0$.

Then, there exists a constant $C=C(D,\gamma,\bar J,k,p,C_p,X_0,q_0)>0$ and a unitary operator $V_{X_0[R]}$ on $\mathcal F(\ell^2(\Lambda))$ supported on $X_0[R]$ and commuting with $n_{X_0[R]}$ such that for all bounded observables $O$ on on $\mathcal F(\ell^2(\Lambda))$ supported on $X_0$ satisfying Definition \ref{defn:q0} with $q_0$, we have
\begin{equation}
\label{ineq:thm:lrbtr}
\left\|\rho\brr{O(t)-V^\dagger_{X_0[R]}OV_{X_0[R]}}\right\|_1\leq C\|O\| 
 \left(\frac{v_{1}(t) t}{R}\right)^{p/2-D-1}
\end{equation}
for all $ R\geq 1$ and all $t\in[1,T]$.
\end{theorem}

When $T=\infty$, we interpret $ [0,T]$ as $[0,\infty)$.
Notice that the constant $C$ depends on the initial state only through the $p$ and $C_p$ that appear in \eqref{basic_assump_moment_func}. We will obtain the trace norm estimate in Theorem \ref{thm:maininformal} from Theorem \ref{thm:lrbtr} by bounding $C_p$ uniformly in the initial state for translation-invariant Hamiltonian with strong local repulsion with the same value of $p$.

For the LRB on quantum-mechanical expectation values $\tr\brr{ \rho(\cdot) }$, we require the moment bound \eqref{basic_assump_moment_func} only under the weaker assumption $p>D+1$ and we obtain a further enhanced velocity $v_{\mathrm{ex}}(t)=t^{\frac{D}{p-D-1}}$.

\begin{theorem}[Further--enhanced LRB for expectation values] \label{thm:lrbav} Let $H$ be of the form \eqref{def:Ham} with $W$ satisfying Assumption \ref{ass:V}.
Assume that there exist $T\in [0,\infty]$ and $p>D+1$ such that \eqref{basic_assump_moment_func} holds.
Fix $X_0\subset \Lambda$ and $q_0\geq 0$.

Then, there exists a constant $C=C(D,\gamma,\bar J,k,p,C_p,X_0,q_0)>0$ and a unitary $V_{X_0[R]}$ on $\mathcal F(\ell^2(\Lambda))$ supported on $X_0[R]$ and commuting with $n_{X_0[R]}$  such that for all bounded observables $O$ on $\mathcal F(\ell^2(\Lambda))$ supported on $X_0$ satisfying Definition \ref{defn:q0} with $q_0$, we have
\begin{equation}
\label{main_ineq:thm:lrbav}
\left|\mathrm{tr}\left(\rho\brr{O(t)-V^\dagger_{X_0[R]}OV_{X_0[R]}}\right)\right|\leq C\|O\|\brr{\frac{v_{\mathrm{ex}}(t) t}{R}}^{p-D-1}
\end{equation}
for all $ R\geq 1$ and all $t\in[1,T]$.
\end{theorem}

The key feature of the unitary $V_{X_0[R]}$ is that it is locally supported on $X_0[R]$ and so the results approximate the full time evolution by conjugation with a \textit{local} unitary. This is sufficient to control the commutators that appear in the traditional formulation of the LRB, in particular in Theorem \ref{thm:maininformal}, because $[V^\dagger_{X_0[R]}OV_{X_0[R]},\tilde O]=0$  when $d(O,\tilde O)>R$ and so
\[
[O(t),\tilde O(t)]
=
[O(t)-V^\dagger_{X_0[R]}OV_{X_0[R]},\tilde O].
\]
 See the proof of Theorem \ref{thm:maininformal} for further details. We decided to state the version with the unitary in Theorems \ref{thm:lrbtr} and \ref{thm:lrbav} because it yields the additional information that $V_{X_0[R]}$ commutes  with $n_{X_0[R]}$. For readers interested in an explicit formula for $V_{X_0[r]}$, we note that this can be obtained from the proof, but it is slightly complicated because of various uses of the interaction picture and the concatenation of short-time LRBs. In \cite{kuwahara2022optimal}, it is possible to replace $V_{X_0[r]}$ post-hoc with $e^{-\mathrm{i}tH_{X_0[R]}}$, which is a particularly nice local unitary, but this is not possible in our setting because  $H_{X_0[R]}$ lacks translation symmetry. 




\subsection{Proof of Theorem \ref{thm:maininformal} assuming Theorems \ref{thm:lrbtr} and \ref{thm:lrbav}}
\label{ssect:verify}
 Theorem \ref{thm:maininformal} is now a straightforward application of these results with $T=\infty$. It is obtained by verifying the moment bounds in the presence of (i) strong repulsion and (ii) translation invariance. 

We consider the discrete torus of side length $L\geq 1$,
\[
\Lambda_L=(\mathbb Z/L\mathbb Z)^D.
\]
We work on the bosonic Fock space $\mathcal F(\ell^2(\Lambda_L))$ and consider a Bose-Hubbard type Hamiltonian of the form
\begin{equation}
    \label{eq:mainexampleHamiltonian}
H=H_0+W,\qquad H_0=-J\sum_{i\sim j } (b_i^\dagger b_j+b_j^\dagger b_i)
,\quad W=\sum_{i\in \Lambda_L}\left[n_i^p+  \tilde w(n_i)\right].
\end{equation}
There exist $c_{\tilde w},\epsilon>0$ such that $| \tilde w(n)|\leq  c_{\tilde w} (n+1)^{p-\epsilon}$ for all $n\geq 0$. 

Our goal is to apply Theorems \ref{thm:lrbtr} and \ref{thm:lrbav} with $R=d(O,\tilde O)+1$ and $T=\infty$.
To verify the conditions of Theorems \ref{thm:lrbtr} and \ref{thm:lrbav}, it suffices to prove the moment bound
$
\sup_{i\in\Lambda_L}\tr\brr{ \rho n_i^p(t) } \le C_p$ for all $t\geq 0.$
We focus on the case $p>2D+2$ as the case $p>D+1$ is handled exactly in the same way.

Once we have verified the moment bound, then we obtain from  Theorem \ref{thm:lrbtr} and H\"older's inequality for Schatten spaces
\[  
 \left\|\rho([O(t),\tilde O])\right\|_1
 = \left\|\rho([O(t)-V^\dagger_{X_0[R]}O V_{X_0[R]},\tilde O])\right\|_1
 \leq C\|O\|\|\tilde O\| \left(\frac{v_1(t) t}{d(O,\tilde O)}\right)^{p/2-D-1}
\]
and similarly for the expectation value. Therefore, the proof of Theorem \ref{thm:maininformal} is completed by the following lemma.

\begin{lemma}[Moment bound]\label{lm:moment}
    Let $p>1$ and let $\rho$ be a translation-invariant quantum state satisfying \[
    \frac{\tr\brr{ \rho H }}{|\Lambda_L|}=E_\rho<\infty.
    \] There exists $C=C(J,p,D,\epsilon,c_v)$ such that for any $L\geq 2$, we have
    \begin{equation}
    \label{main_ineq_lm:moment}
          \sup_{i\in\Lambda_L}  \tr\brr{ \rho n_i^p(t) }\leq 
    2\br{E_\rho+C},\qquad \textnormal{for all }t\geq 0.
    \end{equation}
\end{lemma}

Explicitly,
\begin{equation}\label{eq:explicitly}
    C(J,p,D,\epsilon,c_{\tilde w})=\max_{n\in \mathbb N}
    \left(-\frac{n^p}{2}+c_{\tilde w}(n+1)^{p-\epsilon}+4JDn\right).
\end{equation}

\begin{proof}[Proof of Lemma \ref{lm:moment}]
    We first control the hopping term $H_0$ by the repulsion. By the Cauchy-Schwarz inequality for operators, we have the operator inequality
    \[
    H_0\geq - J\sum_{i\sim j}(n_i+n_j)=-4JDN.
    \]
  Hence, using also that $| \tilde w(n)|\leq  c_{\tilde w} (n+1)^{p-\epsilon}$, we have
    \[
    \begin{aligned}
    H=H_0+V \geq & \sum_{j\in \Lambda_L} \br{ n_j^p+\tilde w (n_j)-4JDn_j}\\
     \geq&  \sum_{j\in \Lambda_L} \br{ n_j^p-c_{\tilde w} (n_j+1)^{p-\epsilon}-4JDn_j}.
    \end{aligned}
    \]
 By energy conservation (more precisely, the fact that the Heisenberg dynamics preserves the Hamiltonian) and cyclicity of the trace, we have
    \[
    \begin{aligned}
     \tr\brr{ \rho H }
    =  \tr\brr{ \rho(-t) H }
    \geq\sum_{j\in \Lambda_L}  \tr\brr{ \rho(-t) \br{n_{j}^p -c_{\tilde w} (n_j+1)^{p-\epsilon}-4JD n_{j}}}.
    \end{aligned}
    \]
      Fix a point $i\in\Lambda$. Since $H$ commutes with translations, the time-evolved state $\rho(-t)=e^{-\mathrm{i}tH}\rho e^{\mathrm{i}tH}$ is also translation-invariant. This yields
      \[
         \sum_{i\in \Lambda_L}  \tr\brr{ \rho(-t) \br{n_j^p -c_{\tilde w} (n_j+1)^{p-\epsilon}-4JD n_{j}}}
         =|\Lambda_L| \ \tr\brr{ \rho(-t)\br{n_{i}^p -c_{\tilde w} (n_i+1)^{p-\epsilon}-4JD n_{i}}}.\\
      \]
  Since $p>1$, a simple case distinction shows that 
  \[
  n^p-c_{\tilde w} (n+1)^{p-\epsilon}-4JD n\geq \frac{n^p}{2}-C,\qquad \textnormal{for all }n\geq0,
  \]
  with $C=C(J,p,D,\epsilon,c_{\tilde w})$. By cyclicity of the trace and $\tr(\rho)=1$,  this yields
  \[
   \tr\brr{ \rho H }\geq |\Lambda_L|\tr\brrr{ \rho(-t)\brr{\frac{1}{2}n_{i}^p -C}}=|\Lambda_L|\br{\frac{1}{2}\tr\brr{ \rho n_{i}(t)^p }-C}.
  \]
Since $ \frac{\tr\brr{ \rho H }}{|\Lambda_L|}=E_\rho$, this
  proves Lemma \ref{lm:moment} and therefore also Theorem \ref{thm:maininformal}.
  \end{proof}

\subsection{An improved particle propagation bound for all moments}
The main result 1 of \cite{kuwahara2022optimal} can be stated as the particle propagation bound  
\begin{equation}\label{eq:known2}
\tr\brr{\rho n_i^q(t)}^{1/q}\leq C \br{t^{D}+q+t},
\end{equation}
for any $q\geq 1$ and $i\in\Lambda$, under the assumption that the initial state satisfies $\tr\brr{\rho n_i^q}^{1/q} \leq C q$.

It turns out that our techniques allow to improve this bound for general moments $q$ by suitably interpolating between \eqref{eq:known2} and Lemma \ref{lm:moment} in the setting of Theorem \ref{thm:maininformal}. This shows that the rigidity effect we leverage also constrains transport properties of higher moments of the local particle number.

\begin{prop}[Improved particle propagation bound]\label{prop:interpolation}
Let $q\geq p$ and suppose that the initial state satisfies $\tr\brr{\rho n_i^q}^{1/q} \leq C_0 q$. Under the assumptions of Theorem \ref{thm:maininformal}, there exists $C=C(J,p,q,C_q, D,\epsilon,c_v,E_\rho)$
such that we have the bound
\begin{equation}\label{eq:improvedprop}
\br{\tr\brr{\rho n_i(t)^{q}}}^{1/q}
\leq  C t^{D(1-p/q)},\qquad t\geq 1.
\end{equation}
\end{prop}

Compared to \eqref{eq:known2}, this has the additional negative term $-Dp/q$, which is a substantial improvement when $p/q$ is of order-$1$. The improvement becomes progressively weaker as $q\to\infty$ and the bound continuously turns into Lemma \ref{lm:moment} at the endpoint $q=p$. These features are expected, because we prove Proposition \ref{prop:interpolation} by interpolation.

\begin{proof}[Proof of Proposition \ref{prop:interpolation}]
 We use the Riesz-Thorin interpolation theorem for a suitably chosen classical probability space.
       For $q\geq p$, let $q_1\geq 1$ and $\theta\in [0,1]$ be such that
  \[
\frac{1}{q}=\frac{1-\theta}{p}+\frac{\theta}{q_1}.
\]
Then, we claim that
\begin{equation}\label{eq:interpolationbound}
\tr\brr{\rho n_i(t)^{q}}^{1/q}\leq C \br{t^{D}+q_1+t}^{\theta}.
\end{equation}
 Given the state $\rho(-t)$, consider its spectral decomposition
  \[
  \rho(-t)=\sum_{j=0}^\infty \lambda_j \Pi_{\psi_j},
  \]
  where $\Pi_{\psi_j}$ project onto the eigenvector $\psi_j$ corresponding to the eigenvalue $\lambda_j$. We allow for degeneracy and choose the $\{\psi_j\}_j$  orthonormal. 

  Fix a site $i\in \Lambda$. Define a classical probability distribution $\mathfrak p:\mathbb N\to [0,1]$ by
  \[
  \mathfrak p(n)=\sum_{j=0}^\infty \lambda_j \|\Pi _{n_i=n} \psi_j\|^2
  \]
  (Notice that indeed $\sum_{n\geq 0} \mathfrak p(n)=\sum_{j}\lambda_j=1$).

  Then expectation values of powers of $n_i$ can be viewed as $L^q$-norms over this probability space 
  \[
  \br{\tr \brr{\rho n_i^q}}^{1/q}=  \mathbb E[X^q]^{1/q}\equiv \|X\|_q,
  \]
 where $X$ denotes the corresponding $\mathbb N$-valued classical random variable.

 In this setting, we can apply the Riesz-Thorin interpolation theorem for general $L^q$-measure spaces,  which says that
 \[
 \|X\|_{q_\theta}\leq \|X\|_{2}^{1-\theta} \|X\|_{q}^{\theta}
 \]
 with $\theta,q_\theta$ as above. Writing this out in terms of the quantum state $\rho(-t)$ and applying Lemma \ref{lm:moment} and \eqref{eq:known2} afterwards yields
   \[
    \tr\brr{\rho n_i^{q_\theta}(t)}^{1/q_\theta}
    =\tr\brr{\rho(-t) n_i^{q_\theta}}^{1/q_\theta}
    \leq \tr\brr{\rho n_i^2(t)}^{\frac{1-\theta}{2}} 
    \tr\brr{\rho n_i^q(t)}^{\frac{\theta}{q}} 
\leq C \br{t^{D}+q}^\theta.
    \] 
    This proves \eqref{eq:interpolationbound}.

Notice that the bound \eqref{eq:interpolationbound} improves over \eqref{eq:known2} for every $q$ because the right-hand side of \eqref{eq:interpolationbound} behaves as $t^{D\theta}$ for $t\gg 1$ and  $\theta<1$. It remains to choose the parameters suitably.
 We fix $q>p$ and find a pair of $\theta\in [0,1]$ and $q_1>q$ such that $\frac{1}{q}=\frac{1-\theta}{p}+\frac{\theta}{q_1}$. The minimal value of the factor $\theta$ is achieved for $q_1\to \infty$ (in which case $\theta\to 1-p/q$). However, this choice makes the right-hand side of~\eqref{eq:interpolationbound} infinitely large. We can remedy this by making the (effectively very large) choice $q_1=t^D$. This gives
\begin{equation}\label{theta_value_1}
\theta=\frac{1-p/q }{1-p/t^D}=\br{1-p/q_*} \br{1 + \frac{p}{t^D-p}}  ,
\end{equation}
and so \eqref{eq:interpolationbound} becomes
\begin{equation}\label{eq:interpolationbound__0}
\tr\brr{\rho n_i(t)^{q}} \leq C t^{D(q-p) \br{1 + \frac{p}{t^D-p}}}
\leq  C t^{D(q-p)},
\end{equation}
where the second bounds holds because $t^{D(q-p)\frac{p}{t^D-p}}=\exp\brr{D(q-p)p\frac{\log(t)}{t^D-p}}=\orderof{1}$ as $t\to\infty$. 
This proves Proposition \ref{prop:interpolation}.
\end{proof}

\begin{remark}
The attentive reader may wonder if Proposition \ref{prop:interpolation} can be used in place of Lemma \ref{lm:moment} in the proof of the Lieb-Robinson bound to further improve the light cone beyond what we establish in Theorems \ref{thm:lrbtr} and \ref{thm:lrbav}. This, however, is not the case.
Implementing such a strategy in the arguments detailed below would ultimately have the effect that a parameter that was independent of time before, is changed to a quantity proportional to $t^{D(q-p)}$, and hence the light cone shape for would be given by
\begin{equation}\label{eq:interpolationbound__1}
\norm{ \rho \brr{ O_X(t)-O_{i_0[R]}^{(\bar{m})} }  }_1 \lesssim 
t^{D(q-p)/2} \frac{t^{q/2-1}}{R^{q/2-D-1}}
\end{equation}
for any $q>p$.
Similarly, for the LRB on the quantum expectation case, we would have 
\begin{equation}\label{eq:interpolationbound__2}
\abs{ \tr \brr{\rho\br{ O_X(t)-O_{i_0[R]}^{(\bar{m})} } } } \lesssim 
t^{D(q-p)} \frac{t^{q-1}}{R^{q-D-1}}
\end{equation}
for any $q>p$.
By optimizing these bounds in $q> p$, one finds that depending on the value of $p$, these bounds reproduce either the velocity bound $\sim t^D$ from \cite{kuwahara2022optimal} or $v_1(t)$ for \eqref{eq:interpolationbound__1} and $v_{\mathrm{ex}}(t)$ for \eqref{eq:interpolationbound__2} that we prove in Theorem \ref{thm:maininformal}. The reason for this is that the light cone exponent is a ratio of linear functions in $q$ and therefore monotonic, which means it reaches its minimum always at one of the endpoints $q=p$ or $q=\infty$. In summary, the a posteriori Riesz-Thorin interpolation does not pay off from the perspective of light cone scaling, only for particle transport.
\end{remark}

\section{Proof strategy for Theorems~\ref{thm:lrbtr} and \ref{thm:lrbav}}\label{sect:strategy}
\subsection{Terminology and notation}
We occasionally refer to all constants that are independent of $R,t$ and the number of vertices or edges appearing in the graph $\Lambda$ as ``$\mathcal O(1)$ constants'' for simplicity of presentation. Examples of such $\mathcal O(1)$ constants are the bound on the hopping strengths $\bar J$, the graph-geometric parameters $\gamma$ and $D$ from Assumption \ref{ass:graph} and the interaction range $k$ from Assumption \ref{ass:V}. 

We recall that the full Hamiltonian on $\Lambda$ is of the form given in \eqref{def:Ham} and \eqref{eq:Wdef}, i.e.,
\[
\begin{aligned}
H=& H_0+ W,\\
\textnormal{with } H_0=&  \sum_{\substack{i,j \in \Lambda:\\ i\sim j}} J_{i,j} (b_i b_j^\dagger +b_i^\dagger b_j ),\qquad W=\sum_{Z\subset \Lambda} w_Z(\{n_i\}_{i\in X}).
\end{aligned}
\]
where we have set $w_Z=0$ for $\mathrm{diam}(Z)> k$ in line with Assumption \ref{ass:V}.

In order to suitably decompose the Hamiltonian, we introduce for an arbitrary subset $X \subseteq \Lambda$, the associated subsystem Hamiltonians $H_{0,X}$, $W_X$ and $H_X$ as follows:
\begin{align}
&H_X= H_{0,X}+ W_X , \notag \\
&H_{0,X}\coloneqq  \sum_{\substack{i,j \in X:\\ i\sim j}}  J_{i,j} (b_i b_j^\dagger +{\rm h.c.} ), \quad   W_X\coloneqq   \sum_{Z \subseteq X} w_Z .
 \label{def:Ham_subset}
\end{align}
Note that these subsystem Hamiltonians are supported on the subset $X$.
Also relevant for decomposing the Hamiltonian are the boundary hopping terms between $X$ and its complement, which we denote as follows:
\begin{align}
\partial h_{X}\coloneqq  H_0-H_{0,X} - H_{0,X^\co}=\sum_{i \in \partial X} \sum_{j\in X^\co: \dist_{i,j}=1} J_{i,j} (b_i b_j^\dagger +{\rm h.c.} ) .
 \label{def:Ham_surface}
\end{align}
(Even though the interaction $W$ also has terms connecting $X$ and $X^c$, these can be treated differently via the interaction picture, because they are mutually commuting.)\\

\subsection{The short-time LRB and how it implies Theorem~\ref{thm:lrbtr} } For the remainder of this section, we work under the assumptions of Theorem~\ref{thm:lrbtr}. 
 
As part of our overall proof strategy, we use the by now standard unitary-concatenation technique that goes back to \cite{kuwahara2016exponential}; see also \cite{kuwahara2021absence,kuwahara2021lieb,kuwahara2022optimal}. This technique allows to reduce the LRB to short-time LRBs (see Theorem \ref{Theorem_for_small_time_evo} below) by a concatenation procedure that we explain now. The short-time LRB is formulated in terms of two parameters $\tau,r>0$ and it is about localizing the time evolution for a timer duration $\tau$ inside a region of spatial extent $r$. In the concatenation procedure, these play the following roles: $\tau$ is an $\mathcal O(1)$ parameter that corresponds to the duration of a single time step which then gets concatenated $t/\tau$ many times. The spatial parameter $r$ will be chosen $\approx R/t$ and since we are considering a situation in which the velocity grows with time, $r$ ends up being a large parameter.

\begin{Theorem}[Short-time LRB] \label{Theorem_for_small_time_evo}
There exists $\tau_0=\tau_0(\bar J,\gamma,D,k)\in (0,1)$ such that the following holds. Consider any subset $X\subset \Lambda$ and parameters $\tau,r>0$ subject to the constraints
\begin{align}
  \label{basic_conditions_tau_r}
\tau \le \tau_0,  \quad e^{-r/4} |X| \le 1 .
\end{align} 
There exists a unitary operator $U_{X[r]}$ supported on $X[r]$ and commuting with $n_{X[r]}$ and a constant $C=C(\bar J,\gamma,D,k,p,C_p,q_0)>0$ such that the following bound holds for any quantum state $\rho$ satisfying
\begin{align}\label{eq:rhoconditionst}
\tr\brr{ \tilde\rho n_i^p(t) } \le C_p,\qquad \textnormal{for all }t\in [0,\tau],
\end{align}
and any bounded observable $O_X$ satisfying Definition \ref{defn:q0} with $q_0$, it holds that
\begin{align}
&\norm{\tilde\rho\brr{ O_{X}(\tau) - U_{X[r]}^\dagger O_{X} U_{X[r]}  }  }_1  
\le C \|O_X\|  \br{  |\partial (X[r])| r e^{-r/(4k)} + |X[r]\setminus X| r^{-p/2+1}} .
\label{Theorem_for_small_time_evo_main_ineq}
\end{align} 
\end{Theorem}

In essence, the statement is that conjugation by $U_{X[r]}$ serves as an appropriate local approximation of the Heisenberg time evolution associated to $e^{-\mathrm{i}H\tau}$. We prove Theorem \ref{Theorem_for_small_time_evo} in Section \ref{sect:Theorem_for_small_time_evo}. The explicit form of $U_{X[r]}$ is not important for what follows; it suffices that it exists, commutes with the particle number, and that it  serves as a local approximation of the time evolution in the sense of Theorem \ref{Theorem_for_small_time_evo}. We obtain its existence through Proposition \ref{prop:short_time_Lieb-Robinson} which is taken from \cite{kuwahara2021lieb}, from which an explicit can be obtained if desired.

\begin{proof}[Proof of Theorem~\ref{thm:lrbtr} assuming Theorem~\ref{Theorem_for_small_time_evo}]
 Consider the right-hand side of \eqref{ineq:thm:lrbtr} as a function of $t\in [1,T]$ and $R\geq 1$. Notice that the left-hand side can be bounded trivially by $2\|O\|$ by the triangle inequality and unitarity. This has two consequences: First, the claim of Theorem~\ref{thm:lrbtr}  is trivial when $R/(v_1(t) t)= R/t^\zeta\leq 1$ where we set $\zeta=1+\frac{D-1}{p/2-D-1}>1$. Second, because we assume $t\geq 1$ in Theorem~\ref{thm:lrbtr}, the claim is  is also trivial if $R$ is bounded by some $R_0$ that depends only on $\mathcal O(1)$ parameters $ D,\gamma,\bar J,k,p,C_p,X_0,q_0$. Therefore, we may assume without loss of generality that $R>t^\zeta$ and that $R>R_0=R_0(D,\gamma,\bar J,k,p,C_p,X_0,q_0)$ and we will use these assumptions below.\\ 

\underline{Step 1:} \textit{Verifying the conditions for the short-time LRB.}\\
The main idea is to apply Theorem~\ref{Theorem_for_small_time_evo} various times and concatenate the resulting short-time LRBs.

Let $t\in [1,T]$ and $R\geq 1$. We define $s$ as the largest positive number $\leq \tau_0$ such that $t/\tau$ is an integer, i.e., we set
\begin{equation}\label{eq:revtau}
\tau:= \frac{t }{\lceil t/\tau_0\rceil }.
\end{equation}
Note that
\[
\tau\leq \tau_0 \quad\textnormal{ and }\quad \bar m:=\frac{t}{\tau}=\lceil t/\tau_0\rceil\in\mathbb Z_+.
\]
We can then decompose the time interval $[0,t]$ into $\bar m$ time steps of length $\tau$. We will use the short-time LRB on each time interval.

It remains to successively decompose the spatial regions in a matching way.  It is convenient to measure distance relative to a single site. For this, we find a lattice site $i_0$ and radius $r_0$ such that the initial region $X_0\subset i_0[r_0]$, i.e., it is contained in the ball of radius $r_0$ centered at $i_0$.

Then, for each $0\leq j\leq \bar m$, we define the subset $X_j$ as follows: 
\begin{align}
\label{Choice_Delta_r_X_m}
X_j:= i_0[r_0+ j\Delta r], \quad r:= \left \lfloor \frac{R-r_0}{\bar{m}} \right \rfloor.
\end{align}
In particular, $X_{\bar{m}} \subseteq i_0[R]\subset X_0[R]$. 

Our goal is to apply Theorem~\ref{Theorem_for_small_time_evo} once for each $j\in\{1,\ldots,\bar m\}$ to obtain $\bar m$ unitaries $\{U_{X_j}\}_{j=1}^{\bar m}$. 
Let  $j\in\{1,\ldots,\bar m\}$. We choose the parameters suitably and check the conditions of Theorem~\ref{Theorem_for_small_time_evo}. First of all, $\tau\leq \tau_0$ and $r$ are given by \eqref{eq:revtau}
and \eqref{Choice_Delta_r_X_m}, respectively. We take $X=X_{j-1}$ and $ \tilde\rho=\rho((j-\bar m)\Delta t)$. Condition \eqref{eq:rhoconditionst} follows from \eqref{basic_assump_moment_func}.

The main condition to check is the second condition in \eqref{basic_conditions_tau_r}, i.e., $e^{-R\tau/(4t)}  |X_{j-1}|\le 1$. Using $X_j\subset X_{\bar m}\subset i_0[R]$  and our geometric Assumption \ref{ass:graph}, this is implied by 
\begin{equation}\label{eq:revRcond}
\gamma e^{-R\tau /(4t)} R^D\leq 1.
\end{equation}
We estimate
\begin{equation}\label{eq:tautest}
\frac{\tau}{t}=\frac{1 }{\lceil t/\tau_0\rceil }\geq \frac{\tau_0}{t+\tau_0}\geq \frac{1}{2} \min\left\{1,\frac{\tau_0}{t}\right\}.
\end{equation}
We distinguish two cases. If $\tfrac{\tau}{t}\geq\tfrac{1}{2}$, then \eqref{eq:revRcond} can be ensured by assuming that $R>R_0(D,\gamma)$ holds which can be done without loss of generality as described in the beginning of the proof.
If $\tfrac{\tau}{t}\geq\tfrac{\tau_0}{2t}$, then \eqref{eq:revRcond} we recall that we may also assume without loss of generality that $R>t^\zeta$ with $\zeta>1$. From $\tfrac{\tau}{t}\geq\tfrac{\tau_0}{2t}$ and this, we obtain
\[
\gamma e^{-R\tau/ (4t)} R^D\leq \gamma e^{-\tfrac{1}{8}R^{1-1/\zeta}\tau_0} R^D.
\]
which holds for $R>R_0(D,\gamma,\tau_0)=R_0(D,\gamma,\bar J,\gamma,D,k)$ where we used that $\tau_0=\tau_0(\bar J,\gamma,D,k)>0$. This verifies the conditions of Theorem~\ref{Theorem_for_small_time_evo} with the above parameter choices for all $j\in\{1,\ldots,\bar m\}$.\\ 

\underline{Step 2:} \textit{Applying the short-time LRB.}\\
From Step 1, we have obtained unitaries $\{U_{X_j}\}_{j=1}^{\bar m}$ such that the conclusion of Theorem~\ref{Theorem_for_small_time_evo} holds for every $j\in\{1,\ldots,\bar m\}$ with $X=X_{j-1}$ and $\tau,r$ as in Step 1.  
We now define the unitary
\[
V_{X[R]}:=U_{X_1}\ldots U_{X_{\bar m}}.
\]
Note that this is a unitary that acts only on $X_0[R]$ because $X_j\subseteq X_{\bar m}\subseteq X_0[R]$ by construction.

In Theorem~\ref{thm:lrbtr}, we consider an arbitrary bounded observable $O$ supported on $X_0$ and satisfying Definition \ref{defn:q0} with a fixed $q_0$. We iteratively define local approximations $\brrr{O_{X_{j}}^{(j)} }_{j=1}^{\bar m}$ via
\begin{align}
\label{Choice_O_X_j_j_Delta/t}
O_{X_{j}}^{(j)} := U_{X_j}^\dagger O_{X_{j-1}}^{(j-1)} U_{X_j}\for 1\leq j\leq \bar{m},\qquad O_{X_{0}}^{(0)}:=O.
\end{align} 
Since each $U_{X_j}$ commutes with particle number, each $O_{X_{j}}^{(j)}$ also satisfies Definition \ref{defn:q0} with the same $q_0$. 
Therefore, we can use Theorem~\ref{Theorem_for_small_time_evo} with the parameter choices from Step 1 on $O_X\equiv O_{X_{j-1}}^{(j-1)}$ for each $j\in\{1,\ldots,\bar m\}$.
This gives
\begin{align}\label{eq:nowapplyst1}
\norm{ \rho((j-\bar m)\tau)\br{ O^{(j)}_{X_{j-1}}(\tau) - U_{X_j}^\dagger O_{X_{j-1}}^{(j)} U_{X_j}  } }_1
\leq C\|O\|
 \br{ |\partial X_j| r e^{- r/(4k)} + |X_j\setminus X_{j-1}| r^{-p/2+1}}.
\end{align} 
where we also used that $\|O_{X_{j}}^{(j)}\|=\|O\|$ by unitarity.
We can successively use the automorphism property of the Heisenberg dynamics, $(AB)(s)=A(s)B(s)$, to express the left-hand side of \eqref{ineq:thm:lrbtr} as the following telescopic sum,
\begin{align}
\label{Lieb_Robinson_connection_ineq}
\left\|\rho\br{O(t)-V^\dagger_{X_0[R]}OV_{X_0[R]}}\right\|_1
=\norm{
\sum_{j=1}^{\bar{m}} 
\rho((j-\bar m)\tau)\br{ O_{X_{j-1}}^{(j-1)}(\tau) - O_{X_j}^{(j)} }  }_1.
\end{align} 
By the triangle inequality and \eqref{eq:nowapplyst1}
\begin{align}
&\left\|\rho\br{O(t)-V^\dagger_{X_0[R]}OV_{X_0[R]}}\right\|_1\\
&\le C\|O\| \sum_{j=1}^{\bar m}
 \br{ |\partial X_j| r e^{-r/(4k)} + |X_j\setminus X_{j-1}| r^{-p/2+1}}\\
 &\le C\|O\| \bar m
 \br{R^{D-1}r e^{-r/(4k)} +R^{D-1} r^{-p/2+2}}
\end{align} 
In the second step, we used Assumption \ref{ass:graph} to bound $|\partial X_j|\leq \gamma (r_0+j r)^{D-1}\leq \gamma R^{D-1}$ and to bound 
\[
|X_j\setminus X_{j-1}|
\leq \gamma \sum_{r_0+(j-1)r<\ell\leq r_0 + jr}\ell^{D-1}\leq C r (r_0+jr)^{D-1}\leq C r R^{D-1},
\]
where the second-to-last inequality follows by comparing with an integral.

To relate the prefactor $\bar m$ to $t$, we use \eqref{eq:tautest} and $\tau_0<1\leq t$ to estimate $\tau\geq \tfrac{1}{2}\min\{t,\tau_0\}\geq \tfrac{\tau_0}{2}$. This implies $\bar m =\frac{t}{\tau}\leq t \frac{2}{\tau_0}$ and so
\[
\left\|\rho\br{O(t)-V^\dagger_{X_0[R]}OV_{X_0[R]}}\right\|_1
\leq C\|O\| t
 \br{R^{D-1}r e^{-r/(4k)} +R^{D-1} r^{-p/2+2}}
\]
Recall that $r= \left \lfloor \frac{R-r_0}{\bar{m}}\right\rfloor$. Since we can assume without loss of generality that $R$ is large compared to $r_0$ and since $\bar m$ is comparable to $t$ up to $\mathcal O(1)$ constants, we have $r\leq C \frac{R}{t}$ and $r^{-1}\leq C' \frac{R}{t}$ for suitable $\mathcal O(1)$-constants $C,C'>0$.

Absorbing all constants into the prefactor $C$, we find
\[
\left\|\rho\br{O(t)-V^\dagger_{X_0[R]}OV_{X_0[R]}}\right\|_1\leq C\|O\|\br{R^D e^{-R/(4kt)} 
+ \left(\frac{v_{1}(t) t}{R}\right)^{p/2-D-1}}.
\]
It remains to remove the first term $R^D e^{-R/(4kt)}$. To this end, we argue similarly as above, recalling that we may assume $R/t^\zeta>1$ with $\zeta>1$. Using this as well as $t\geq 1$, we obtain
\[
R^D e^{-R/(4kt)}\leq R^D e^{-R^{1-1/\zeta}/(4k)}
\leq (t^{\zeta}/R)^{p/2-D-1} \underbrace{R^{2D+1-p/2}e^{-(R^{1-1/\zeta})/(4k)}}_{\leq C'=C'(D,p,k)}
\leq C' (t^{\zeta}/R)^{p/2-D-1} 
\]
and so
\[
\left\|\rho\br{O(t)-V^\dagger_{X_0[R]}OV_{X_0[R]}}\right\|_1\leq C  \|O\|\left(\frac{v_{1}(t) t}{R}\right)^{p/2-D-1}.
\]
This proves Theorem~\ref{thm:lrbtr} with the unitary $V_{X_0[R]}=\prod_{j=1}^{\bar m} U_{X_j}$.
\end{proof}

\subsection{Short-time LRB for expectation values  and how it implies Theorem~\ref{thm:lrbav} }
The overall proof strategy is the same as that for Theorem~\ref{thm:lrbtr}. That is, we concatenate suitable short-time LRBs which now control the average.
That is, our main goal is to derive an analogous statement to Theorem~\ref{Theorem_for_small_time_evo} for the expectation value.

%

\begin{Theorem}[Short-time LRB for the expectation] \label{Theorem_for_small_time_evo_average}
Under the same setup as in Theorem~\ref{Theorem_for_small_time_evo}, there exists a unitary  $U_{X[r]}$ supported on $X[r]$ and a constant $C=C(\bar J,\gamma,D,k,p,C_p,q_0)>0$ such that the following bound holds for any quantum state $\rho$ satisfying
\begin{align}\label{eq:rhoconditionst_av}
\tr\brr{ \rho n_i^p(t) } \le C_p,\qquad \textnormal{for all }t\in [0,\tau],
\end{align}
and any bounded observable $O_X$ satisfying Definition \ref{defn:q0} with $q_0$, it holds that\begin{align}
&\abs{\tr \brr{\br{ O_{X}(\tau) - U_{X[r]}^\dagger O_{X} U_{X[r]}  } \rho} }
\le C\|O_X\| \br{ |\partial (X[r])| r e^{-r/(4k)} +  |X[r]\setminus X| r^{-p+1}}.
\label{Theorem_for_small_time_evo_main_ineq_average}
\end{align}  
\end{Theorem}

Notice that, compared to Theorem~\ref{Theorem_for_small_time_evo}, the last power of $r$ is $-p+1$ instead of $-p/2+1$. This is the reason why the time dependence in Theorem \ref{thm:lrbav} is better than in \ref{thm:lrbtr}.

We prove Theorem \ref{Theorem_for_small_time_evo_average} in Section \ref{sect:Theorem_for_small_time_evo_average}.

\begin{proof}[Proof of Theorem \ref{thm:lrbav} assuming Theorem~\ref{Theorem_for_small_time_evo_average}] This proof closely follows the line of argumentation by which we derived Theorem \ref{thm:lrbtr} from Theorem~\ref{Theorem_for_small_time_evo}. In particular, the verification of the conditions in Step 1 is identical. In Step 2, we use the same definitions of $V_{X_0[R]}$ and $O_{X_j}^{(j)}$. The telescopic sum from \eqref{Lieb_Robinson_connection_ineq} is replaced with
\begin{align}
\label{Lieb_Robinson_connection_ineq_ave}
\abs{\tr\brr{ \br{ O(t)-V^\dagger_{X_0[R]}OV_{X_0[R]}} }}
= \abs{ \sum_{j=1}^{\bar{m}} \tr{ \brr{ \rho((j-\bar m)\Delta t)\br{ O_{X_{j-1}}^{(j-1)}(\Delta t) - O_{X_j}^{(j)} } }}} .
\end{align} 
The rest of the argument is completely analogous to the proof of Theorem \ref{thm:lrbtr}, with the exception that $r^{-p/2+1}$ is replaced by $r^{-p+1}$ in each application of Theorem~\ref{Theorem_for_small_time_evo_average} compared to Theorem~\ref{Theorem_for_small_time_evo}. The details are left to the reader.
\end{proof}

\section{Proof of Theorem~\ref{Theorem_for_small_time_evo}}
\label{sect:Theorem_for_small_time_evo}\label{sect:st1}
In this section, we work in the setting of Theorem~\ref{Theorem_for_small_time_evo}.

The overall strategy is to approximate the dynamics with a locally truncated dynamics in which the local boson number is cut off at level $\bar q$, which has to be suitably chosen. The approximation error is controlled by Markov's inequality, and the truncated dynamics effectively locally describes a quantum spin system with bounded interactions, depending on $\bar q$, and therefore enjoys Lieb-Robinson bounds depending on $\bar q$. Note that one cannot implement the particle number cutoff globally, since the corresponding union bound would produce a prefactor growing with $\Lambda$, which would be detrimental, so one instead cuts off the dynamics only on a suitably chosen annulus; see Figure \ref{fig_Lieb-Robinson_setup}.
Compared to previous works, Ref.~\cite{kuwahara2021lieb,kuwahara2022optimal}, the technical novelty here is that the decay in $\bar q$ is only polynomial with a fixed $p$.
This point forces us to more carefully analyze the error estimation due to the truncation of the boson number at each of the sites.  
In the proof, we also rely on several technical results from Ref.~\cite{kuwahara2022optimal} which we quote as needed. Our presentation is focused on the steps that require novel ideas.

\subsection{Setting up the proof}

\begin{definition}
 For any  $\tilde X\subseteq \Lambda$ and $q\in \mathbb{N}$,  we define the spectral projection
 \begin{align} \label{def:bar_Pi_X_z}
\bar{\Pi}_{\tilde X,q} := \prod_{i\in \tilde X}\Pi_{n_i\le q}, 
\end{align}
i.e., $\bar{\Pi}_{\tilde X,q}$ ($\tilde X\subseteq \Lambda$) is the spectral projection onto the subspace such that for an arbitrary $i \in \tilde \Lambda$ the boson number $ n_i$ is truncated at level $q$.

We define an effective Hamiltonian $\tilde{H}[\tilde X,q]$ by conjugating the full Hamiltonian with $\bar{\Pi}_{\tilde X,q}$,
 \begin{align}
 \label{def:tilde_Ham_effective}
&\tilde{H}[\tilde X,q] := \tilde{H}_0[\tilde X,q] + \tilde{V} [\tilde X,q] ,\notag \\
& \tilde{H}_0[\tilde X,q]:=\bar{\Pi}_{\tilde X,q} H_{0} \bar{\Pi}_{\tilde X,q} ,\quad 
\tilde V[\tilde X,q]:=\bar{\Pi}_{\tilde X,q} W \bar{\Pi}_{\tilde X,q} ,
\end{align}
\end{definition}

We fix a region $X$ and parameters $\tau,r>0$ subject to the constraints \eqref{basic_conditions_tau_r}.
In the proof, we consider the subset $\tilde{X}\subseteq \Lambda$ defined as 
 \begin{align}
\tilde{X}:= X[r] \setminus X[r/2] .
\end{align}

 \begin{figure}[tt]
\centering
\includegraphics[clip, scale=0.4]{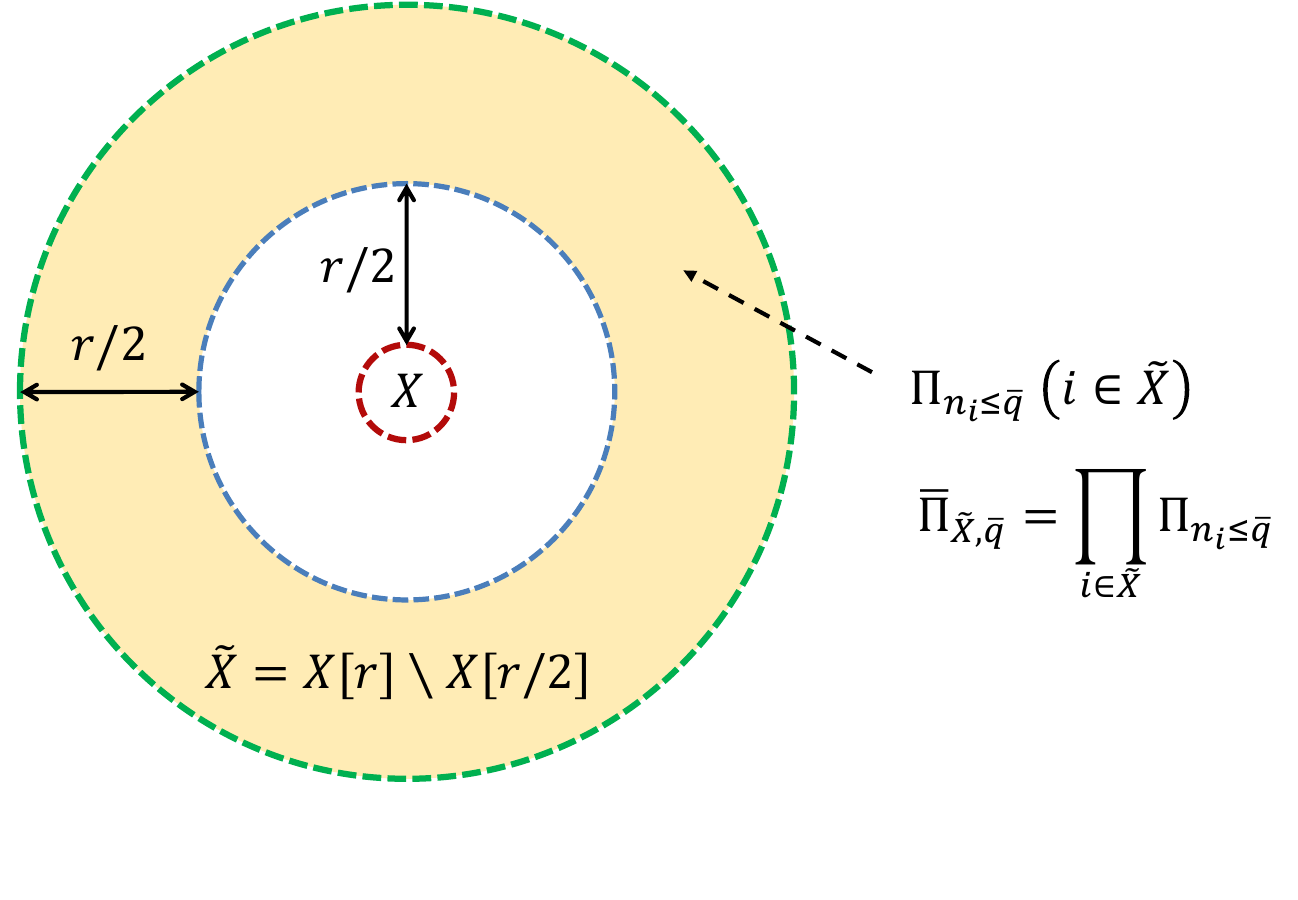}
\caption{Schematic of the construction of the effective Hamiltonian $\tilde{H}[\tilde{X},\bar{q}] = \bar{\Pi}_{\tilde{X},\bar{q}} H \bar{\Pi}_{\tilde{X},\bar{q}}$.
The spectral projection $\bar{\Pi}_{\tilde{X},\bar{q}}$ implements the truncation of the local boson number at level $\bar{q}$ in the region of $\tilde{X}=X[r]\setminus X[r/2]$. 
}
\label{fig_Lieb-Robinson_setup}
\end{figure}
In the following, it is convenient to make the Hamiltonian explicit in the notation for the Heisenberg dynamics, i.e., 
\[
O(H,t):=e^{\mathrm{i}tH}Oe^{-\mathrm{i}tH}.
\]

Roughly speaking, our approach to proving Theorem \ref{Theorem_for_small_time_evo} is to approximate the dynamics $O_{X}(H,\tau)$ in two steps as follows. First, we consider the truncated dynamics and show that it satisfies a suitably LRB
 \begin{align}
 \label{second_approx_LR_0}
O_{X}(\tilde{H}[\tilde{X},\bar{q}] , \tau)  \approx  V_{X[r]}^\dagger O_{X} V_{X[r]},
\end{align}
where $ V_{X[r]} $ is the desired unitary operator. Second, we approximate the full dynamics with the truncated dynamics
 \begin{align}
  \label{first_approx_LR_0}
O_{X}(H,\tau) \rho \approx O_{X}(\tilde{H}[\tilde{X},\bar{q}] , \tau) \rho, 
\end{align}
for quantum states $\rho$ satisfying the union bound.

For the first approximation \eqref{second_approx_LR_0}, we can directly utilize the previous analyses as follows. 

%

\begin{prop} [(S.97) in Ref.~\cite{kuwahara2021lieb}]\label{prop:short_time_Lieb-Robinson}
We work under the assumptions of Theorem \ref{Theorem_for_small_time_evo} and we choose $\bar{q}=r/2$. There exists $\tau_0=\tau_0(\bar J,\gamma,D,k)\in (0,1)$ such that the following holds for all $\tau\leq \tau_0$.

There exists a constant $C=C(\bar J,\gamma,D,k)>0$ and a unitary  $U_{X[r]}$ supported on $X[r]$ and commuting with $n_{X[r]}$ such that for every bounded observable $O_X$ supported on $X$,
 \begin{align}
 \label{ineq:prop:short_time_Lieb-Robinson}
\left\| O_{X}(\tilde{H}[\tilde{X}, r/2] , \tau)-U_{X[r]}^\dagger O_{X} U_{X[r]} \right\| \le 
C \|O_X\|   |\partial (X[r])|r e^{-r/(2k)}.
\end{align}
\end{prop}

We remark that the explicit construction of the unitary operator $U_{X[r]}$ is given in Ref.~\cite[Section S.V.C]{kuwahara2021lieb}. In fact, $\tau_0$ can be defined explicitly as 
\begin{equation}\label{eq:tau0defn}
\tau_0:=(2^6 e^2 \bar J \gamma^3 k (2k)^{2D})^{-1},
\end{equation}
where $e$ is Euler's constant, $\bar J$ is the bound on hopping matrix, and $\gamma,D$ bounds the surface growth in the graph geometry (cf.\ Assumption \ref{ass:graph}), but the explicit formula is not needed for the argument.

The remaining problem is  to estimate the approximation error for~\eqref{first_approx_LR_0}.

\begin{prop}\label{prop:eff_Ham_approx}
We work under the assumptions of Theorem \ref{Theorem_for_small_time_evo} and we choose $\bar{q}=r/2$. Let $\tau_0$ be given by \eqref{eq:tau0defn}. Then, the following holds for all $\tau\leq \tau_0$.

Let $\rho$ be an arbitrary quantum state on $\mathcal F(\ell^2(\Lambda))$ satisfying the moment bound
\begin{align}
\tr \br{  n_i^p \rho } \le C_p .
\end{align} 
Assume further that 
\begin{align}
e^{-r/4} |X| \le 1.
\end{align} 
Then, there exists a constant $C=C(\bar J,\gamma,D,k,p,C_p,q_0)$ such that for every bounded observable $O_X$  satisfying  Definition \ref{defn:q0} with $q_0$, 

\begin{align}
\norm{ O_X(H,\tau)  \rho - O_X(\bar{\Pi}_{\tilde{X},\bar{q}} H\bar{\Pi}_{\tilde{X},\bar{q}}, \tau)  \rho}_1 
\leq
C \|O_X\| |\tilde{X}| r^{-p/2+1}
\label{main_ineq_prop:eff_Ham_approx} 
\end{align}  
\end{prop}

\subsection{Proof of Theorem~\ref{Theorem_for_small_time_evo} from two key propositions}

\begin{proof}[Proof of Theorem~\ref{Theorem_for_small_time_evo} assuming Proposition~\ref{prop:eff_Ham_approx}]
Set $\bar{q}= r/2$.
By Proposition \ref{prop:short_time_Lieb-Robinson}, specifically \eqref{ineq:prop:short_time_Lieb-Robinson},
\[
\left\|\br{ O_{X}(\tilde{H}[\tilde{X}, r/2] , \tau)-U_{X[r]}^\dagger O_{X} U_{X[r]}}\tilde\rho \right\|_1 \le 
C\|O_X\|   |\partial (X[r])| re^{-r/(2k)}
\]
for a suitable constant $C=C(\bar J,\gamma,D,k)$.
We apply Proposition~\ref{prop:eff_Ham_approx} to $\rho=\tilde\rho$ and \eqref{main_ineq_prop:eff_Ham_approx} gives
\[
\norm{ O_X(H,\tau)  \rho - O_X(\bar{\Pi}_{\tilde{X},\bar{q}} H\bar{\Pi}_{\tilde{X},\bar{q}}, \tau)  \rho}_1 
\leq C \|O_X\| |\tilde{X}| r^{-p/2+1}.
\]
for a suitable constant $C=C(\bar J,\gamma,D,k,p,C_p,q_0)$.

From the triangle inequality for the trace norm and $\tilde{X}\subset X[r]\setminus X$,  we then obtain~\eqref{Theorem_for_small_time_evo_main_ineq}. This proves Theorem~\ref{Theorem_for_small_time_evo} assuming Proposition~\ref{prop:eff_Ham_approx}.
\end{proof}

In the remainder of this section, we prove Proposition~\ref{prop:eff_Ham_approx}.

\subsection{Proof of Proposition~\ref{prop:eff_Ham_approx}}

We first state the two lemmas that we combine to obtain Proposition~\ref{prop:eff_Ham_approx}.

The first lemma is a general, abstract bound on the difference between the full Heisenberg dynamics
and the dynamics generated by an effective projected Hamiltonian. We write
 $\norm{\ldots}_F$ for the Frobenius norm, i.e., 
 \[
 \norm{O}_F:= \sqrt{\tr(O O^\dagger)}.
 \]

\begin{lemma}[Corollary~26 in Ref.~\cite{kuwahara2022optimal}]  \label{corol:effective Hamiltonian_accuracy}
Consider a separable Hilbert space $\mathcal H$. Let $\rho$ be a quantum state on $\mathcal H$ and let $\bar{\Pi}$ be an orthogonal projection on $\mathcal H$.
We have, \begin{align}
\label{error_time/evo_effectve_original_ineq_mix}
\norm{ O(H,\tau)  \rho - O(\bar{\Pi} H \bar{\Pi}, \tau)  \rho}_1 &\le  \norm{\bar{\Pi}^\co O(H,\tau) \sqrt{\rho} }_F+
\norm{\bar{\Pi}^\co O \sqrt{\rho(\tau)}}_F
+  \norm{\bar{\Pi}^\co \sqrt{\rho} }_F
+ \norm{\bar{\Pi}^\co \sqrt{\rho(\tau)} }_F  \notag \\
&+\int_0^\tau \norm{ \bar{\Pi}H_0 \bar{\Pi}^\co O(H,\tau-\tau_1) \sqrt{\rho(\tau_1) }  }_F  d\tau_1
+\int_0^\tau \norm{ \bar{\Pi}H_0 \bar{\Pi}^\co \sqrt{\rho(\tau_1)} }_F  d\tau_1 ,
\end{align}  

\end{lemma}

\begin{proof}[Proof of Lemma \ref{corol:effective Hamiltonian_accuracy}]
This follows from the Duhamel formula and the triangle inequality; see also the proof of Corollary~26  in \cite{kuwahara2022optimal}.
\end{proof}

In the remainder of Section 4, we work under the assumptions of Proposition~\ref{prop:eff_Ham_approx}. Let $\tau_0$ be given by \eqref{eq:tau0defn}. Consider any $\tilde X,\bar q$ as in the statement of Proposition~\ref{prop:eff_Ham_approx}. 

When applying \eqref{error_time/evo_effectve_original_ineq_mix} with $H$ the Hamiltonian and $\bar{\Pi}= \tilde{\Pi}_{\tilde{X},\bar{q}}$ we see that all the terms in the right-hand side of the inequality~\eqref{error_time/evo_effectve_original_ineq_mix} can be expressed through (integrals of) norms of the form
\begin{align}
\label{01_lema:bar_Pi_upper_bound_general_F}
\norm{ \bar{\Pi}_{\tilde{X},\bar{q}}^\co O_{X} (\tau_2) \sqrt{\rho(\tau_1)} }_F ,
\end{align} 
and 
\begin{align}
\label{02_lema:bar_Pi_upper_bound_general_F}
\norm{ \bar{\Pi}_{\tilde{X},\bar{q}}H_{0,Y} \bar{\Pi}_{\tilde{X},\bar{q}}^\co O_{X} (\tau_2) \sqrt{\rho(\tau_1)} }_F   .
\end{align} 
for suitable choices of $\tau_1,\tau_2\leq \tau$. For example, the first term on the right-hand side of the inequality~\eqref{error_time/evo_effectve_original_ineq_mix} is of the form \eqref{01_lema:bar_Pi_upper_bound_general_F} with $\tau_2=\tau$ and $\tau_1=0$ and the second term on the right-hand side of the inequality~\eqref{error_time/evo_effectve_original_ineq_mix} is of the form \eqref{01_lema:bar_Pi_upper_bound_general_F}  with $\tau_2=0$ and $\tau_1=\tau$, etc.

The second lemma bounds expressions of the form \eqref{01_lema:bar_Pi_upper_bound_general_F} and \eqref{02_lema:bar_Pi_upper_bound_general_F}.

\begin{lemma}\label{lemm:upper_bound_effective_terms}
Suppose the assumption of Proposition~\ref{prop:eff_Ham_approx} hold. There exists a constant $C=C(\bar J,\gamma,D,k,p,C_p,q_0)$ such that for every $\tau_1,\tau_2\leq \tau$ and for every bounded observable $O_X$ satisfying  Definition \ref{defn:q0} with $q_0$, 
\begin{align}
\label{ineq_O_X_bar_Pi_tau_2_1_upp_1st}
\norm{ \bar{\Pi}_{\tilde{X},\bar{q}}^\co O_X (\tau_2) \sqrt{\rho(\tau_1)} }_F\le 
C \|O_X\| |\tilde{X}|^{1/2}  r^{-p/2}
 \end{align}
and 
\begin{align}
\label{ineq_O_X_bar_Pi_tau_2_2_upp_1st}
\norm{ \bar{\Pi}_{\tilde{X},\bar{q}}H_0 \bar{\Pi}_{\tilde{X},\bar{q}}^\co O_X (\tau_2) \sqrt{\rho(\tau_1)} }_F   
\le C \|O_X\|  |\tilde X| r^{1-p/2}
\end{align}
\end{lemma}

This result can be viewed as a polynomial analog of Proposition~29 in  Ref.~\cite{kuwahara2022optimal}.

\begin{proof}[Proof of Proposition \ref{prop:eff_Ham_approx} assuming Lemma \ref{lemm:upper_bound_effective_terms}]
We first apply Lemma \ref{corol:effective Hamiltonian_accuracy}. Afterwards, we apply Lemma \ref{corol:effective Hamiltonian_accuracy} with the suitable choices of $\tau_1$ and $\tau_2$ on the right-hand side in~\eqref{error_time/evo_effectve_original_ineq_mix}. This implies that there exists a constant $C=C(\bar J,\gamma,D,k,p,C_p,q_0)$ such that
\[
\norm{ O_X(\tau)  \rho - O_X(\bar{\Pi}_{\tilde{X},\bar{q}} H\bar{\Pi}_{\tilde{X},\bar{q}},\tau)  \rho }_1 
\leq 
C |\tilde{X}|^{1/2}  r^{-p/2}\br{1+C r |\tilde X|}
\leq C |\tilde{X}|  r^{1-p/2},
\]
%
%
which is the desired inequality~\eqref{main_ineq_prop:eff_Ham_approx}. This proves Proposition \ref{prop:eff_Ham_approx} assuming Lemma \ref{lemm:upper_bound_effective_terms}.
\end{proof}

It therefore remains for us to prove Lemma \ref{lemm:upper_bound_effective_terms}.

\subsubsection{Generalized moment bounds}
 We will control the spectral projector in Lemma \ref{lemm:upper_bound_effective_terms} by moments of the local particle number through Markov's inequality. The moments are then controlled by the following bounds on traces of moments of the local particle number times $O_X (\tau_2) \rho(\tau_1) \brr{O_X (\tau_2) }^\dagger$, i.e., 
\begin{align}\label{eq:littlerewrite}
\tr\brr{ O_X (\tau_2) \rho(\tau_1) \brr{O_X (\tau_2) }^\dagger   n_i^p} 
= \tr\brr{  e^{\mathrm{i}H\tau_2} O_X  \rho(\tau_1-\tau_2) O_X^\dagger e^{-\mathrm{i}H\tau_2}   n_i^p}. 
\end{align} 
\begin{lemma} \label{lemm:moment_upper_bound_O_X}
Let $\rho$ be an arbitrary quantum state satisfying the moment upper bound as 
\begin{align}\label{eq:momentUBprop}
\sup_{0\leq \tau\leq \tau_0}\sup_{i\in\Lambda}\tr \br{  n_i(\tau)^p \rho } \le C_p .
\end{align} 
Then, there exists a constant $C=C(\bar J,\gamma,D,k,p,C_p,q_0)$ such that for any $i\in\Lambda$, any $0\leq \tau_2\leq\tau_1  \leq \tau_0$, and every observable $O_X$ satisfying  Definition \ref{defn:q0} with $q_0$, 
\begin{align}
\tr\brr{  e^{\mathrm{i}H\tau_2} O_X  \rho(\tau_1-\tau_2) O_X^\dagger e^{-\mathrm{i}H\tau_2}   n_i^p}
\le 
C\|O_X\|^2 \br{1+|X| e^{-\tfrac{3}{4}d_{i,X}}}
 \label{lemm:moment_upper_bound_O_X_main_ineq}
\end{align} 
\end{lemma}

We will apply this for $i\in \tilde X$ and use the standing assumption $e^{-r/2}|X|\leq 1$ to bound $|X| e^{-d_{i,X}}$.

\subsubsection{Proof of Lemma~\ref{lemm:moment_upper_bound_O_X}}

Fix any $i\in\Lambda$.  We introduce \begin{align}
\label{Eq:def_tilde_D_X}
\hat{\mathcal{D}}_i:=  \sum_{j\in \Lambda} e^{-\tfrac{3}{4}\dist_{i,j}}  n_j  .
\end{align}
By applying Theorem~1 in Ref.~\cite{kuwahara2022optimal} with a suitable $R=R(\bar J,\gamma,D,k,p,C_p,q_0)$, there exists $C_1=C_1(\bar J,\gamma,D,k,p,C_p,q_0)$ such that we have the operator inequality
\[
 n_i(-\tau_2)^p\leq C_1\br{1+ \hat{\mathcal{D}}_i}^p.
\]
We note that Theorem~1 in Ref.~\cite{kuwahara2022optimal} uses the fact that $\tau_0 \le 1/(4\gamma \bar{J})$, while follows from \eqref{eq:tau0defn} (or which could otherwise be trivially satisfied by decreasing $\tau_0$).

We abbreviate $\rho':=\rho(\tau_1-\tau_2)$ in the remainder of this proof. 
We have
\begin{align}
&\tr\brr{  e^{\mathrm{i}H\tau_2} O_X  \rho(\tau_1-\tau_2) O_X^\dagger e^{-\mathrm{i}H\tau_2}   n_i^p}\\
&=\tr\brr{  e^{\mathrm{i}H\tau_2} O_X \rho' O_X^\dagger e^{-\mathrm{i}H\tau_2}   n_i^p} \notag\\
&=  \tr\brrr{  O_X \rho' O_X^\dagger \brr{ n_i(-\tau_2)}^p} \notag \\
&\le C_1 \tr\brr{  O_X \rho' O_X^\dagger \br{ 1+ \hat{\mathcal{D}}_i }^p}.
\label{p1_lemm:moment_upper_bound_O_X}
\end{align} 
 Next, we recall without proof the following elementary result.

\begin{lemma}[Lemma~3 in Ref.~\cite{kuwahara2022optimal}]\label{lem:boson_operator_norm}
Let $A_{q_0}$ be a bounded operator satisfying Definition \ref{defn:q0} with $q_0$. Let $\{\nu_j\}_{j\in \Lambda}$ be a collection of non-negative numbers and set  $\bar{\nu}:= \max_{j\in  \tilde\Lambda} (\nu_j)$.
Then, we have 
\begin{align}
\label{main_ineq_lem:boson_operator_norm}
A_{q_0}^\dagger \br{\sum_{j\in \Lambda} \nu_j  n_j}^m A_{q_0} 
\leq 4^{m}\norm{A_{q_0}}^{2} \br{4 \bar{\nu}q_0^3 + \bar{\nu} n_{ \tilde\Lambda} + \sum_{j\notin  \tilde\Lambda} \nu_j  n_j }^m  ,\qquad \textnormal{for any }\tilde\Lambda\subseteq \Lambda.
\end{align}
\end{lemma}

We apply Lemma \ref{lem:boson_operator_norm}
 to the last line in \eqref{p1_lemm:moment_upper_bound_O_X}
 with the choices 
\begin{align}
A_q= O_X, \quad \nu_j =e^{-\tfrac{3}{4}\dist_{i,j}}  ,\quad \bar{\nu} =  e^{-\tfrac{3}{4}\dist_{i,X}} , \quad m= p,\quad \tilde\Lambda= X, 
\end{align}
We thus find that there exists $C=C(\bar J,\gamma,D,k,p,C_p,q_0)$ such that
\[
O_X^\dagger \br{1+  \hat{\mathcal{D}}_i}^p O_X \leq 
C\|O_X\|^2
 \br{ 1 +  e^{-\tfrac{3}{4}\dist_{i,X}}  n_X + \sum_{j\notin X} e^{-\tfrac{3}{4}\dist_{i,j}}  n_j }^p.
 \]
To summarize,
\begin{align}
\tr\brr{  e^{\mathrm{i}H\tau_2} O_X  \rho' O_X^\dagger e^{-\mathrm{i}H\tau_2}   n_i^p}
\leq C\|O_X\|^2
\tr\brr{\rho' \br{ 1 +  e^{-\tfrac{3}{4}\dist_{i,X}}  n_X + \sum_{j\notin X} e^{-\tfrac{3}{4}\dist_{i,j}}  n_j }^p}.
\label{p2_lemm:moment_upper_bound_O_X}
\end{align} 

We now discuss how to control the right-hand side of this expression.
Recall that $\rho'=\rho(\tau_1-\tau_2)$ and $\tau_1-\tau_2\in [0,\tau_0]$. Let $m\le p$. By H\"older's inequality, which can be lifted to an operator inequality for commuting operators, and the moment upper bound \eqref{eq:momentUBprop},
\begin{align}
\tr\brr{ \rho' \prod_{s=1}^m  n_s } \le  \prod_{s=1}^m  \tr\brr{ \rho' n_s^m }^{1/m} \le \br{C_m^{1/m}}^m =C_m .
\end{align} 
By Jensen's inequality, $C_m^{1/m} \le C_p^{1/p}$ for any $m\le p$. Hence, 
\begin{equation}\label{p3_lemm:moment_upper_bound_O_X}
\tr\brr{\rho' \br{ 1 +  e^{-\tfrac{3}{4}\dist_{i,X}}  n_X + \sum_{j\notin X} e^{-\tfrac{3}{4}\dist_{i,j}}  n_j }^p}
\leq \tr\brr{\rho' \br{ 1 +  e^{-\tfrac{3}{4}\dist_{i,X}}  |X| C_p^{1/p} +C(\gamma,D) C_p^{1/p}  }^p}
\end{equation}
%
%
where we used 
\begin{align}
\sum_{j\notin X} e^{-\tfrac{3}{4}\dist_{i,j}} \le 1+ \sum_{s=1}^\infty \sum_{j:j\in \partial i[s]} e^{-\tfrac{3}{4}s}
\le C(\gamma,D).
\end{align} 
The claim follows from \eqref{p3_lemm:moment_upper_bound_O_X}. 
This completes the proof. 
\qed

\subsubsection{Proof of Lemma~\ref{lemm:upper_bound_effective_terms}}
The proof is a polynomial analog of the one in Ref.~\cite[Proposition~29]{kuwahara2022optimal}. 
We choose $r$ such that $e^{-r/4} |X| \le 1$. Recalling also \eqref{eq:littlerewrite}, we can apply \eqref{lemm:moment_upper_bound_O_X_main_ineq} from Lemma \ref{lemm:moment_upper_bound_O_X}.

\begin{align}
\sup_{i\in\tilde X}\tr\brr{ O_X (\tau_2) \rho(\tau_1) \brr{O_X (\tau_2) }^\dagger   n_i^p} \le C \|O_X\|^2 ,\label{lemm:moment_upper_bound_O_X_main_ineq/apply}
\end{align} 
where we used that $d_{i,X}\geq r/2$  for every $i\in \tilde X$.
Hence, by Markov's inequality, for every $0\leq s\leq p$,
\begin{align}
\label{prob_upper_bound_e^-s_ineq_1}
\tr\brr{ n_i^s \Pi_{n_i > \bar q} O_X (\tau_2) \rho(\tau_1)  O_X (\tau_2)^\dagger } 
=\tr\brr{ n_i^{s-p}\Pi_{n_i > \bar q}  n_i^{p}  O_X (\tau_2) \rho(\tau_1)  O_X (\tau_2)^\dagger   } 
\le C \|O_X\|^2  \bar{q}^{s-p} .
\end{align}
By using the above inequality with $s=0$ and  a union bound:
\begin{align}
\label{ineq_O_X_bar_Pi_tau_2_1_upp_derive}
\norm{ \bar{\Pi}_{\tilde{X},\bar{q}}^\co O_X (\tau_2) \sqrt{\rho(\tau_1)} }_F^2=\tr \brr{ \bar{\Pi}_{\tilde{X},\bar{q}}^\co O_X (\tau_2) \rho(\tau_1)  O_X (\tau_2)^\dagger \bar{\Pi}_{\tilde{X},\bar{q}}^\co }   
&\le \sum_{i\in \tilde{X}}\tr \brr{  \Pi_{n_i > \bar{q}} O_X (\tau_2) \rho(\tau_1)  O_X (\tau_2)^\dagger  }   \notag \\
&\le  C \|O_X\|^2   |\tilde{X}|  (\bar{q})^{-p} .
\end{align}
Recalling that $\bar q=r/2$ and taking square roots, this proves the first inequality~\eqref{ineq_O_X_bar_Pi_tau_2_1_upp_1st}.

\medskip
To prove the second claimed inequality~\eqref{ineq_O_X_bar_Pi_tau_2_2_upp_1st}, we first decompose
\begin{align}
\label{decomposition/pi_L_H_0_co}
\bar{\Pi}_{\tilde{X},\bar{q}}H_0 \bar{\Pi}_{\tilde{X},\bar{q}}^\co = \bar{\Pi}_{\tilde{X},\bar{q}} \br{H_{0,\tilde{X}}  +  \partial h_{\tilde{X}}} \bar{\Pi}_{\tilde{X},\bar{q}}^\co ,
\end{align}
where we use the definition \eqref{def:Ham_surface} of $\partial h_{\tilde{X}} $ (i.e., $\partial h_{\tilde{X}} = H_0- H_{0,\tilde{X}}-H_{0,\tilde{X}^\co}$)  and observed that $\bar{\Pi}_{\tilde{X},\bar{q}} H_{0,\tilde{X}^\co }\bar{\Pi}_{\tilde{X},\bar{q}}^\co =H_{0,\tilde{X}^\co } \bar{\Pi}_{\tilde{X},\bar{q}} \bar{\Pi}_{\tilde{X},\bar{q}}^\co =0$.

By Cauchy-Schwarz,
\begin{align}
\label{upper_bound_hopping_operator}
| b_i b_j^\dagger| = \sqrt{(b_i b_j^\dagger)^\dagger b_i b_j^\dagger} = \sqrt{ n_i ( n_j+1)} \leq  n_i+ n_j ,
\end{align}
and 
\begin{align}
\label{definition_bar_pai_ij_bar_q}
&\bar{\Pi}_{\tilde{X},\bar{q}} b_i b_j^\dagger \bar{\Pi}_{\tilde{X},\bar{q}}^\co = \bar{\Pi}_{\tilde{X},\bar{q}} b_i b_j^\dagger \bar{\Pi}_{i,j,\bar{q}}^\co ,\quad 
\bar{\Pi}_{i,j,\bar{q}}:=\begin{cases}
 \Pi_{n_i\le \bar{q}}\Pi_{n_j\le \bar{q}} &\for i,j\in \tilde{X}, \\
 \Pi_{n_i\le \bar{q}} &\for i\in \tilde{X} , \quad j\in \tilde{X}^\co,
 \end{cases}
\end{align}
we obtain 
\begin{align}
\label{ineq_O_X_bar_Pi_tau_2_2_upp_0}
\norm{ \bar{\Pi}_{\tilde{X},\bar{q}}H_0 \bar{\Pi}_{\tilde{X},\bar{q}}^\co O_X (\tau_2) \sqrt{\rho(\tau_1)} }_F   
&\le \bar{J}\sum_{i\in \tilde{X}} \sum_{j: \dist_{i,j}=1}  \norm{ \bar{\Pi}_{\tilde{X},\bar{q}} (b_i b_j^\dagger +{\rm h.c.} )\bar{\Pi}_{i,j,\bar{q}}^\co O_X (\tau_2) \sqrt{\rho(\tau_1)} }_F    \notag \\
&\le\bar{J}\sum_{i\in \tilde{X}} \sum_{j: \dist_{i,j}=1}  \norm{\br{|b_i b_j^\dagger| +|b_jb_i^\dagger|}\bar{\Pi}_{i,j,\bar{q}}^\co O_X (\tau_2) \sqrt{\rho(\tau_1)} }_F    \notag \\
&\le2 \bar{J}\sum_{i\in \tilde{X}} \sum_{j: \dist_{i,j}=1}  \norm{( n_i+ n_j)\bar{\Pi}_{i,j,\bar{q}}^\co O_X (\tau_2) \sqrt{\rho(\tau_1)} }_F .
\end{align}
From the inequality~\eqref{prob_upper_bound_e^-s_ineq_1} with $s=2$, we obtain
\begin{align}
\label{ineq_O_X_bar_Pi_tau_2_2_upp_02}
\tr\brr{( n_i + n_j)^2 \bar{\Pi}_{i,j,\bar{q}}^\co O_X (\tau_2) \rho(\tau_1)O_X (\tau_2)^\dagger  } 
&=\tr\brr{( n_i + n_j)^{2-p} \bar{\Pi}_{i,j,\bar{q}}^\co( n_i + n_j)^{p}  O_X (\tau_2) \rho(\tau_1)O_X (\tau_2)^\dagger  } \notag\\
&\le \bar{q}^{2-p} \tr\brr{( n_i + n_j)^{p}  O_X (\tau_2) \rho(\tau_1)O_X (\tau_2)^\dagger  } \notag\\
&\le C\|O_X\|^2 \bar{q}^{2-p} 
\end{align}
where we also again used the H\"older inequality $\norm{\rho  n_i^s  n_j^{p-s}}_1 \le \norm{\rho  n_i^p}_1^{s/p}\norm{\rho  n_j^p}^{1-s/p}_1$ that holds for an arbitrary quantum state $\rho$.

By combining the inequalities~\eqref{ineq_O_X_bar_Pi_tau_2_2_upp_0} and \eqref{ineq_O_X_bar_Pi_tau_2_2_upp_02}, we arrive at the claimed inequality~\eqref{ineq_O_X_bar_Pi_tau_2_2_upp_1st} as follows: 
\begin{align}
\label{ineq_O_X_bar_Pi_tau_2_2_upp_0_fin}
\norm{ \bar{\Pi}_{\tilde{X},\bar{q}}H_0 \bar{\Pi}_{\tilde{X},\bar{q}}^\co O_X (\tau_2) \sqrt{\rho(\tau_1)} }_F   
&\le  C\|O_X\|^2  \sum_{i\in \tilde{X}} \sum_{j: \dist_{i,j}=1}\bar{q}^{1-p/2}   \notag \\
&\le C\|O_X\|^2 \bar{q}^{1-p/2} |\tilde{X}|,
\end{align}
where we also used $\sum_{j: \dist_{i,j}=1} 1 \le |i[1]|\le \gamma$ for all $ i\in \Lambda$.
We thus obtain the second main inequality~\eqref{ineq_O_X_bar_Pi_tau_2_2_upp_1st}. 
This completes the proof of Lemma~\ref{lemm:upper_bound_effective_terms} and hence of Proposition \ref{prop:eff_Ham_approx}. \qed

\section{Proof of Theorem~\ref{Theorem_for_small_time_evo_average}}\label{sect:Theorem_for_small_time_evo_average}\label{sect:st2}

\subsection{Overview of the proof}

For the proof of Theorem~\ref{Theorem_for_small_time_evo_average} we use the same notation as in the proof of Theorem~\ref{Theorem_for_small_time_evo}. The key difference arises from the estimation of the error using the effective Hamiltonian $\bar{\Pi}_{\tilde{X},\bar{q}} H\bar{\Pi}_{\tilde{X},\bar{q}}$, which was previously given by Proposition~\ref{prop:eff_Ham_approx}. 
For the expectation, we do not use the direct analog of Lemma~\ref{corol:effective Hamiltonian_accuracy}, but instead we  use the following variant. We recall that
\[
\norm{O}_F= \sqrt{\tr(O O^\dagger)}
\]
is the Frobenius norm.
\begin{lemma}  \label{lem:effective Hamiltonian_accuracy_Average of the commutator}
Consider a separable Hilbert space $\mathcal H$. Let $\rho$ be a quantum state on $\mathcal H$ and let $\bar{\Pi}$ be an orthogonal projection on $\mathcal H$.
For every bounded operator $O$ on $\mathcal H$ with $\|O\|=1$, we have
\begin{align}
\label{error_time/evo_effectve_original_ineq_mix_Average of the commutator}
&\tr \brr{\br{ O(H,\tau)   - O(\bar{\Pi} H \bar{\Pi}, \tau) } \rho}  \notag \\
\leq&   \norm{\bar{\Pi}^\co \sqrt{\rho}}_F
\br{\norm{ \bar{\Pi}^\co e^{\mathrm{i}H\tau}Oe^{-\mathrm{i}H\tau}\sqrt{\rho}}_F
+
  \norm{\bar{\Pi}^\co Oe^{-\mathrm{i}H\tau}\sqrt{\rho} }_F}
\\
&+\int_0^\tau    \norm{\sqrt{\rho} e^{\mathrm{i}H \tau} O
 e^{-\mathrm{i}H \tau_1}  \bar{\Pi}^{\co}}_F 
 \norm{\bar{\Pi}^{\co} H_0 \bar{\Pi} e^{-\mathrm{i}\tilde{H} (\tau - \tau_1)}   \sqrt{\rho}}_F
 d\tau_1\\
&+2\norm{\bar{\Pi}^{\co}\sqrt{\rho} }_F^2+2\norm{e^{-\mathrm{i}H\tau} \bar{\Pi}^{\co}\sqrt{\rho} }_F^2
 +\br{ \int_0^\tau \norm{\bar{\Pi} H_0 \bar{\Pi}^{\co} e^{-\mathrm{i}{H} (\tau - \tau_1)}\sqrt{\rho} }_F   d\tau_1}^2.
\end{align}  
\end{lemma}

This result is similar to Lemma~\ref{corol:effective Hamiltonian_accuracy} but it leverages cyclicity of the trace to perform Cauchy-Schwarz in a slightly different way that will turn out to be beneficial: Under the integral, the term containing $H_0$ is separated from the term with the observable $O$. This will allow us to improve the $|\tilde X|$-dependence in the bound~\eqref{ineq_O_X_bar_Pi_tau_2_2_upp_1st} via the following lemma. 

There is some finesse involved to avoid terms where $\bar{\Pi}^{\co}$ is evaluated relative to a state of the form$e^{-\mathrm{i}\tilde{H} (\tau - \tau_1)}\sqrt{\rho}$. This would be problematic because controlling $\bar{\Pi}^{\co}$ via Markov's inequality requires the moment bound for the local particle number only  for the dynamics with respect to the full Hamiltonian $H$, not for the dynamics of $\tilde H$. Indeed, in our main example (Theorem \ref{thm:maininformal}) we can only verify the bound for the full Hamiltonian $H$ by using its translation-invariance, which $\tilde H$ fails to have.



\begin{lemma}\label{lemm:upper_bound_effective_terms_average}
We work under the assumptions of Theorem \ref{Theorem_for_small_time_evo_average} and we choose $\bar{q}=r/2$. Let $\tau_0$ be given by \eqref{eq:tau0defn}. Then, the following holds for all $\tau\leq \tau_0$. Let $p>0$ and let $\rho$ be an arbitrary quantum state on $\mathcal F(\ell^2(\Lambda))$ satisfying the moment bound
\begin{align}
\tr \br{  n_i^p \rho } \le C_p .
\end{align} 
Assume further that 
\begin{align}
e^{-r/4} |X| \le 1.
\end{align} 

Then, there exists a constant $C=C(\bar J,\gamma,D,k,p,C_p)$ such that for every $\tau_1\leq \tau$, it holds that
\begin{align}
\label{main_ineq_ave_norm_H_0_eff_ham1}
 \norm{\bar{\Pi}_{\tilde{X},\bar{q}}^\co H_0 \bar{\Pi}_{\tilde{X},\bar{q}} e^{-\mathrm{i}\tilde{H} \tau_1}  \sqrt{\rho}}_F 
\le C |\tilde X|^{1/2} r^{1-p/2}.
\end{align}
\end{lemma}

We postpone the proofs of Lemmas \ref{lem:effective Hamiltonian_accuracy_Average of the commutator} and \ref{lemm:upper_bound_effective_terms_average} for now.

\subsubsection{Proof of Theorem~\ref{Theorem_for_small_time_evo_average} assuming Lemmas \ref{lem:effective Hamiltonian_accuracy_Average of the commutator} and \ref{lemm:upper_bound_effective_terms_average}}
We recall the claim
\[
\abs{\tr \brr{\br{ O_{X}(\tau) - U_{V[r]}^\dagger O_{X} V_{X[r]}  } \rho} }
\le C\|O_X\| \br{ |\partial (X[r])| r e^{-r/(4k)} +  |X[r]\setminus X| r^{-p+1}}.
\]
By Proposition \ref{prop:short_time_Lieb-Robinson}, we already have that there exists $C=C(\bar J,\gamma,D,k)>0$ such that
\[
\left\| O_{X}(\tilde{H}[\tilde{X}, r/2] , \tau)-U_{X[r]}^\dagger O_{X} U_{X[r]} \right\| \le 
C \|O_X\|   |\partial (X[r])| re^{-r/(2k)}.
\]
Therefore, it suffices to prove
\begin{equation}\label{eq:remains_average}
\abs{\tr \brr{\br{ O_{X}(\tau) - O_{X}(\tilde{H}[\tilde{X}, r/2] , \tau) } \rho} }
\leq
C\|O_X\| |X[r]\setminus X| r^{-p+1}
\end{equation}
for a suitable constant $C=C(\bar J,\gamma,D,k,p,C_p)$.

We now prove \eqref{eq:remains_average}. We apply Lemma \ref{lem:effective Hamiltonian_accuracy_Average of the commutator} with the choice $ \bar{\Pi}= \bar{\Pi}_{\tilde{X},\bar{q}}$. We then need to estimate the corresponding terms that appear in the right-hand side of~\eqref{error_time/evo_effectve_original_ineq_mix_Average of the commutator} with $ \bar{\Pi}= \bar{\Pi}_{\tilde{X},\bar{q}}$.
We have already proved in Lemma~\ref{lemm:upper_bound_effective_terms} the following inequality:
\begin{align}
\label{ineq_O_X_bar_Pi_tau_2_1_upp_1st_re}
\norm{ \bar{\Pi}_{\tilde{X},\bar{q}}^\co O_X (\tau_2) \sqrt{\rho(\tau_1)} }_F\le 
C \|O_X\| |\tilde{X}|^{1/2}  r^{-p/2} .
 \end{align}
 Note that Lemma~\ref{lemm:moment_upper_bound_O_X}   also holds for the time evolution with respect to the effective Hamiltonian $\tilde H[\tilde X,\bar q]=\bar{\Pi}_{\tilde{X},\bar{q}} H\bar{\Pi}_{\tilde{X},\bar{q}}$, i.e.,
 \begin{align}
\label{ineq_O_X_bar_Pi_tau_2_1_upp_1st_re-eff}
\norm{ \bar{\Pi}_{\tilde{X},\bar{q}}^\co O_X (\tau_2,\tilde H[\tilde X,\bar q]) \sqrt{\rho(\tau_1)} }_F\le 
C \|O_X\| |\tilde{X}|^{1/2}  r^{-p/2}.
 \end{align}
From Lemma \ref{lemm:upper_bound_effective_terms_average}, we obtain that
\begin{align}
\norm{\bar{\Pi}_{\tilde{X},\bar{q}}^\co H_0 \bar{\Pi}_{\tilde{X},\bar{q}} e^{-\mathrm{i}\tilde{H} \tau_1}  \sqrt{\rho}}_F 
+ \quad \norm{\bar{\Pi}_{\tilde{X},\bar{q}}^\co H_0\bar{\Pi}_{\tilde{X},\bar{q}}  e^{-\mathrm{i}H \tau_1}  \sqrt{\rho}}_F  
\leq C |\tilde{X}|^{1/2}  r^{-p/2}. 
\end{align}  
By applying these inequalities to what we obtain from \eqref{error_time/evo_effectve_original_ineq_mix_Average of the commutator}, we find
 \begin{align}
\label{error_time/evo_effectve_Average}
\tr \brr{\br{ O(H,\tau)  \rho - O(\bar{\Pi}_{\tilde{X},\bar{q}} H \bar{\Pi}_{\tilde{X},\bar{q}}, \tau) } \rho}  &\le  
C |\tilde{X}| r^{-p+1}.
\end{align}  
Recalling that $\tilde X=X[r]\setminus X[r/2]$, we see that this proves \eqref{eq:remains_average} and hence also Theorem \ref{Theorem_for_small_time_evo_average}.
\qed

\subsubsection{Proof of Lemma~\ref{lem:effective Hamiltonian_accuracy_Average of the commutator}}  \label{sec:proof_lem:effective Hamiltonian_accuracy_Average}
We denote 
\begin{align}
\label{start_definition_notation}
\bar{\Pi}^\co =1- \bar{\Pi}, \quad \tilde{H}=\bar{\Pi} H \bar{\Pi}
\end{align}  
for simplicity.
We have \begin{equation}
\begin{aligned}
\label{first_estimation_rho_0_error_0_2}
&\abs{\tr \brr{ \br{ O(H,\tau)  -  O(\tilde{H}, \tau) }\rho}}\\ 
\le&\abs{\tr \brr{\br{ e^{\mathrm{i}H\tau} -  e^{\mathrm{i}\tilde{H}\tau}  }O e^{-\mathrm{i}H\tau}\rho}}
+\abs{\tr \brr{ e^{\mathrm{i}\tilde{H} \tau} O  \br{e^{-\mathrm{i}H\tau} -  e^{-\mathrm{i}\tilde{H} \tau}   }\rho}}\\
\le&\abs{\tr \brr{\br{ e^{\mathrm{i}H\tau} -  e^{\mathrm{i}\tilde{H}\tau}  }O e^{-\mathrm{i}H\tau}\rho}}
+\abs{\tr \brr{ e^{\mathrm{i}H \tau} O  \br{e^{-\mathrm{i}H\tau} -  e^{-\mathrm{i}\tilde{H} \tau}   }\rho}}\\
&+\abs{\tr \brr{  \br{e^{\mathrm{i}H\tau} -  e^{\mathrm{i}\tilde{H} \tau}}O  \br{e^{-\mathrm{i}H\tau} -  e^{-\mathrm{i}\tilde{H} \tau}   }\rho}}\\
=&2\abs{\tr \brr{ e^{\mathrm{i}H \tau} O  \br{e^{-\mathrm{i}H\tau} -  e^{-\mathrm{i}\tilde{H} \tau}   }\rho}}+\abs{\tr \brr{  \br{e^{\mathrm{i}H\tau} -  e^{\mathrm{i}\tilde{H} \tau}}O  \br{e^{-\mathrm{i}H\tau} -  e^{-\mathrm{i}\tilde{H} \tau}   }\rho}}\\
\end{aligned}
\end{equation}
where we used $\overline{\tr{A}}=\tr\brr{A^\dagger}$ and cyclicity in the last step.

Consider the first term for now. We recall Definition \eqref{start_definition_notation}  of $\bar{\Pi}^\co := 1- \bar{\Pi}$ and we decompose
 \begin{equation}\label{calculation_basic_effective_ham_2}
\abs{\tr \brr{ e^{\mathrm{i}H \tau} O  \br{e^{-\mathrm{i}H\tau} -  e^{-\mathrm{i}\tilde{H} \tau}   }\rho}}
\leq
\abs{\tr \brr{ e^{\mathrm{i}H \tau} O \br{e^{-\mathrm{i}H\tau} -  e^{-\mathrm{i}\tilde{H} \tau}   } \bar{\Pi}^\co\rho}}
+
\abs{\tr \brr{ e^{\mathrm{i}H \tau} O \br{e^{-\mathrm{i}H\tau} -  e^{-\mathrm{i}\tilde{H} \tau}   } \bar{\Pi}\rho}}
\end{equation}
For the first term, we use cyclicity, $\bar{\Pi}^\co \tilde{H}=\bar{\Pi}^\co \bar{\Pi} H \bar{\Pi} = 0$, and the Cauchy-Schwarz inequality $\tr(A^\dagger B) \le \norm{A}_F \norm{B}_F$ to obtain
\[
\begin{aligned}
\abs{\tr \brr{  e^{\mathrm{i}H \tau} O \br{e^{-\mathrm{i}H\tau} -  e^{-\mathrm{i}\tilde{H} \tau}   } \bar{\Pi}^\co\rho}}
\leq& \norm{\sqrt{\rho} e^{\mathrm{i}H \tau} O \br{e^{-\mathrm{i}H\tau} -  1}  \bar{\Pi}^\co}_F
\norm{\bar{\Pi}^\co \sqrt{\rho}}_F \\
\leq& \norm{\bar{\Pi}^\co \sqrt{\rho}}_F
\br{\norm{ \bar{\Pi}^\co e^{\mathrm{i}H\tau}Oe^{-\mathrm{i}H\tau}\sqrt{\rho}}_F
+
  \norm{\bar{\Pi}^\co Oe^{-\mathrm{i}H\tau}\sqrt{\rho} }_F}
\end{aligned}
\]
The second term on the right-hand side of \eqref{calculation_basic_effective_ham_2} is treated via Duhamel's formula
\begin{align}
\label{error_term_sum_small_dt}
\br{e^{-\mathrm{i}H\tau}  - e^{-\mathrm{i}\tilde{H}\tau}  }\bar{\Pi}
&=- i  \int_0^\tau e^{-\mathrm{i}H \tau_1}\br{\tilde{H} -H  } \bar{\Pi}e^{-\mathrm{i}\tilde{H} (\tau - \tau_1)}   d\tau_1.
\end{align}  
where we used $[\bar{\Pi},\tilde{H}]=0$. Next, we decompose
\begin{align}\label{eq:Hcomparison}
 \br{\tilde{H}-H  } \bar{\Pi}= \bar{\Pi}H \bar{\Pi}  - \br{\bar{\Pi}+1-\bar{\Pi}} H   \bar{\Pi}
=- \bar{\Pi}^\co H \bar{\Pi}=-\bar{\Pi}^{\co}H_0 \bar{\Pi}  ,
\end{align}  
where we used $[V,\bar{\Pi}]=0$ which is a consequence of $[ n_i, \bar{\Pi}]=0$ for $\forall i\in \Lambda$. 

Therefore, from Eq.~\eqref{error_term_sum_small_dt}, we have  
\begin{align}
\label{error_term_sum_small_dt_reduce}
\abs{\tr \brr{ e^{\mathrm{i}H \tau} O  \br{e^{-\mathrm{i}H\tau} -  e^{-\mathrm{i}\tilde{H} \tau}   } \bar{\Pi}\rho}}
=&\abs{ \int_0^\tau\tr \brr{ e^{\mathrm{i}H \tau} O
 e^{-\mathrm{i}H \tau_1} \bar{\Pi}^{\co} H_0 \bar{\Pi} e^{-\mathrm{i}\tilde{H} (\tau - \tau_1)}  
\rho}  d\tau_1}
 \notag \\
\le&\int_0^\tau \abs{ \tr \brr{ e^{\mathrm{i}H \tau} O
 e^{-\mathrm{i}H \tau_1} \bar{\Pi}^{\co} H_0 \bar{\Pi} e^{-\mathrm{i}\tilde{H} (\tau - \tau_1)}  
\rho}}
 d\tau_1\notag\\ 
\le &
\int_0^\tau    \norm{\sqrt{\rho} e^{\mathrm{i}H \tau} O
 e^{-\mathrm{i}H \tau_1}  \bar{\Pi}^{\co}}_F 
 \norm{\bar{\Pi}^{\co} H_0 \bar{\Pi} e^{-\mathrm{i}\tilde{H} (\tau - \tau_1)}   \sqrt{\rho}}_F
 d\tau_1.
\end{align}  
Returning to \eqref{calculation_basic_effective_ham_2}, we have shown that
\begin{align}
\label{calculation_basic_effective_ham_2_fini}
\abs{\tr \brr{ e^{\mathrm{i}H \tau} O  \br{e^{-\mathrm{i}H\tau} -  e^{-\mathrm{i}\tilde{H} \tau}   }\rho}}
\leq&   \norm{\bar{\Pi}^\co \sqrt{\rho}}_F
\br{\norm{ \bar{\Pi}^\co e^{\mathrm{i}H\tau}Oe^{-\mathrm{i}H\tau}\sqrt{\rho}}_F
+
  \norm{\bar{\Pi}^\co Oe^{-\mathrm{i}H\tau}\sqrt{\rho} }_F}
\\
&+\int_0^\tau    \norm{\sqrt{\rho} e^{\mathrm{i}H \tau} O
 e^{-\mathrm{i}H \tau_1}  \bar{\Pi}^{\co}}_F 
 \norm{\bar{\Pi}^{\co} H_0 \bar{\Pi} e^{-\mathrm{i}\tilde{H} (\tau - \tau_1)}   \sqrt{\rho}}_F
 d\tau_1.
\end{align}

It remains to consider the last term in \eqref{first_estimation_rho_0_error_0_2}. We use $\bar{\Pi}^\co \tilde{H}= 0$, Cauchy-Schwarz, and $[O,\bar{\Pi}]=0$ to obtain
\[
\begin{aligned}
&\abs{\tr \brr{  \br{e^{\mathrm{i}H\tau} -  e^{\mathrm{i}\tilde{H} \tau}}O  \br{e^{-\mathrm{i}H\tau} -  e^{-\mathrm{i}\tilde{H} \tau}   }\rho}}\\
\leq&
\abs{\tr \brr{  \br{e^{\mathrm{i}H\tau} -  1}\bar{\Pi}^{\co}O \bar{\Pi}^{\co} \br{e^{-\mathrm{i}H\tau} -  1   } \rho }}+\abs{\tr \brr{  \br{e^{\mathrm{i}H\tau} -  1}\bar{\Pi} O \bar{\Pi} \br{e^{-\mathrm{i}H\tau} -  1   } \rho }}\\
\leq& \norm{\bar{\Pi}^{\co}\br{e^{-\mathrm{i}H\tau} -  1 }  \sqrt{\rho} }_F^2+\norm{ \bar{\Pi}\br{e^{-\mathrm{i}H\tau} -  e^{-\mathrm{i}\tilde{H} \tau}   }\sqrt{\rho} }_F^2\\
\leq& 2\norm{\bar{\Pi}^{\co}\sqrt{\rho} }_F^2+2\norm{e^{-\mathrm{i}H\tau} \bar{\Pi}^{\co}\sqrt{\rho} }_F^2+\norm{\bar{\Pi}\br{e^{-\mathrm{i}H\tau} -  e^{-\mathrm{i}\tilde{H} \tau}   } \sqrt{\rho} }_F^2
\end{aligned}
\]
For the last term, we recall that Duhamel's formula \eqref{error_term_sum_small_dt} and \eqref{eq:Hcomparison} imply
\[
\bar{\Pi}\br{e^{-\mathrm{i}H\tau}  - e^{-\mathrm{i}\tilde{H}\tau}  }
= i  \int_0^\tau e^{-\mathrm{i}\tilde{H} \tau_1}  \bar{\Pi} H_0 \bar{\Pi}^{\co}e^{-\mathrm{i}{H} (\tau - \tau_1)}   d\tau_1
\]
and so
\[
\norm{\bar{\Pi}\br{e^{-\mathrm{i}H\tau} -  e^{-\mathrm{i}\tilde{H} \tau}   } \sqrt{\rho} }_F\leq 
 \int_0^\tau \norm{\bar{\Pi} H_0 \bar{\Pi}^{\co} e^{-\mathrm{i}{H} (\tau - \tau_1)}\sqrt{\rho} }_F   d\tau_1.
\]
We have thus shown that 
\begin{equation}\label{calculation_basic_effective_ham_fin_2}
\begin{aligned}
&\abs{\tr \brr{  \bar{\Pi}^{\co}\br{e^{\mathrm{i}H\tau} -  1}O  \br{e^{-\mathrm{i}H\tau} -  1   } \bar{\Pi}^{\co}\rho }}\\
\leq&2\norm{\bar{\Pi}^{\co}\sqrt{\rho} }_F^2+2\norm{e^{-\mathrm{i}H\tau} \bar{\Pi}^{\co}\sqrt{\rho} }_F^2
 +\br{ \int_0^\tau \norm{\bar{\Pi} H_0 \bar{\Pi}^{\co} e^{-\mathrm{i}{H} (\tau - \tau_1)}\sqrt{\rho} }_F   d\tau_1}^2.
\end{aligned}
\end{equation}

Therefore, by applying the inequalities~\eqref{calculation_basic_effective_ham_2_fini} and \eqref{calculation_basic_effective_ham_fin_2} to \eqref{first_estimation_rho_0_error_0_2}, we obtain the desired inequality~\eqref{error_time/evo_effectve_original_ineq_mix_Average of the commutator}.
This completes the proof. \qed

\subsubsection{Proof of Lemma~\ref{lemm:upper_bound_effective_terms_average}}
Our task is to estimate the following quantity:
\begin{align}
 \norm{\bar{\Pi}_{\tilde{X},\bar{q}}^\co H_0 \bar{\Pi}_{\tilde{X},\bar{q}} e^{-\mathrm{i}{H} \tau_1}  \sqrt{\rho}}_F^2
 = \tr \brr{\bar{\Pi}_{\tilde{X},\bar{q}} H_0 \bar{\Pi}_{\tilde{X},\bar{q}}^\co H_0 \bar{\Pi}_{\tilde{X},\bar{q}} \rho({H}, \tau_1) }.
\end{align}
By using the relation~\eqref{definition_bar_pai_ij_bar_q}, i.e.,
\begin{align}
\label{definition_bar_pai_ij_bar_q_rre}
&\bar{\Pi}_{\tilde{X},\bar{q}} b_i b_j^\dagger \bar{\Pi}_{\tilde{X},\bar{q}}^\co = \bar{\Pi}_{\tilde{X},\bar{q}} b_i b_j^\dagger \bar{\Pi}_{i,j,\bar{q}}^\co ,\quad 
\bar{\Pi}_{i,j,\bar{q}}:=\begin{cases}
 \Pi_{n_i\le \bar{q}}\Pi_{n_j\le \bar{q}} &\for i,j\in \tilde{X}, \\
 \Pi_{n_i\le \bar{q}} &\for i\in \tilde{X} , \quad j\in \tilde{X}^\co,
 \end{cases}
\end{align}
we obtain
\begin{align}
\label{definition_bar_pai_ij_bar_q_rre_summation_de}
\bar{\Pi}_{\tilde{X},\bar{q}}  H_0  \bar{\Pi}_{\tilde{X},\bar{q}}^\co H_0 \bar{\Pi}_{\tilde{X},\bar{q}}
&= \sum_{i,j: i\in \tilde{X}}\sum_{k,l: k \in \tilde{X}} J_{i,j}J_{k,l}\bar{\Pi}_{\tilde{X},\bar{q}} (b_i b_j^\dagger +{\rm h.c.} )\bar{\Pi}_{\tilde{X},\bar{q}}^\co \cdot \bar{\Pi}_{\tilde{X},\bar{q}}^\co (b_{k} b_{l}^\dagger +{\rm h.c.} ) \bar{\Pi}_{\tilde{X},\bar{q}} \notag \\
&= \sum_{i,j: i\in \tilde{X}}\sum_{k,l: k \in \tilde{X}} J_{i,j}J_{k,l}\bar{\Pi}_{\tilde{X},\bar{q}} (b_i b_j^\dagger +{\rm h.c.} )\bar{\Pi}_{i,j,\bar{q}}^\co \cdot \bar{\Pi}_{k,l,\bar{q}}^\co (b_{k} b_{l}^\dagger +{\rm h.c.} ) \bar{\Pi}_{\tilde{X},\bar{q}}  \notag \\
&=  \sum_{\substack{i,j,k,l\\  i\in \tilde{X},\ k \in \tilde{X},\ \{i,j\}\cap \{k,l\}\neq \emptyset}} J_{i,j}J_{k,l}\bar{\Pi}_{\tilde{X},\bar{q}} (b_i b_j^\dagger +{\rm h.c.} )\bar{\Pi}_{i,j,\bar{q}}^\co \cdot \bar{\Pi}_{k,l,\bar{q}}^\co (b_{k} b_{l}^\dagger +{\rm h.c.} ) \bar{\Pi}_{\tilde{X},\bar{q}} ,
\end{align}
where we use $[\bar{\Pi}_{i,j,\bar{q}}^\co,\bar{\Pi}_{k,l,\bar{q}}^\co ]=0$ and $[b_i b_j^\dagger , \bar{\Pi}_{k,l,\bar{q}}^\co ]=0$ for $\{i,j\}\cap \{k,l\}= \emptyset$. 
Note that $\bar{\Pi}_{\tilde{X},\bar{q}} \bar{\Pi}_{k,l,\bar{q}}^\co=0$. 

Each of the terms in the summation of the right-hand side of~\eqref{definition_bar_pai_ij_bar_q_rre_summation_de} is upper-bounded by
\begin{align}
&\abs{\tr \brrr{\brr{ J_{i,j}J_{k,l}\bar{\Pi}_{\tilde{X},\bar{q}} (b_i b_j^\dagger +{\rm h.c.} )\bar{\Pi}_{i,j,\bar{q}}^\co \cdot \bar{\Pi}_{k,l,\bar{q}}^\co (b_{k} b_{l}^\dagger +{\rm h.c.} ) \bar{\Pi}_{\tilde{X},\bar{q}} }\rho({H}, \tau_1) }} \notag \\
&\le \bar{J}^2 \norm{\bar{\Pi}_{i,j,\bar{q}}^\co (b_i b_j^\dagger +{\rm h.c.} )\sqrt{\rho({H}, \tau_1) } }_F \cdot 
 \norm{\bar{\Pi}_{k,l,\bar{q}}^\co (b_{k} b_{l}^\dagger +{\rm h.c.} )  \sqrt{\rho({H}, \tau_1) } }_F .
 \label{interaction/H__0/upper:boud}
\end{align}
We then consider 
\begin{align}
\norm{\bar{\Pi}_{i,j,\bar{q}}^\co b_i b_j^\dagger \sqrt{\rho({H}, \tau_1) } }_F^2 
&= \tr\brr{ \bar{\Pi}_{i,j,\bar{q}}^\co b_i b_j^\dagger \rho({H}, \tau_1) b_i^\dagger  b_j  } \notag \\
&= \tr\brr{( n_i+ n_j)^{-p+2} \bar{\Pi}_{i,j,\bar{q}}^\co ( n_i+ n_j)^{p} \frac{1}{ n_i+ n_j} b_i b_j^\dagger \rho({H}, \tau_1) b_i^\dagger  b_j \frac{1}{ n_i+ n_j} } \notag \\
&\le \bar{q}^{2-p} \tr\brr{ ( n_i+ n_j)^p \br{\frac{1}{ n_i+ n_j} b_i b_j^\dagger }\rho({H}, \tau_1) \br{\frac{1}{ n_i+ n_j} b_i b_j^\dagger }^\dagger } ,
\label{ineq_bar_pi_bi_bj_2_0}
\end{align}
where we can ensure
\begin{align}
\norm{\frac{1}{ n_i+ n_j} b_i b_j^\dagger} \le 1 
\end{align}
because of $|b_i b_j^\dagger| \leq  n_i+ n_j$ from the inequality~\eqref{upper_bound_hopping_operator}. 
Using Lemma~\ref{lem:boson_operator_norm} by letting
\begin{align}
O_q \to  \br{\frac{1}{ n_i+ n_j} b_i b_j^\dagger }^\dagger , \quad q\to 1, \quad \sum_{j\in \Lambda} \nu_j  n_j \to  n_i+ n_j , \quad  \bar{\nu}\to 1, \quad  \tilde\Lambda\to \{i,j\}
\end{align}
in the inequality~\eqref{main_ineq_lem:boson_operator_norm}, we obtain
\begin{align}
\tr\brr{ ( n_i+ n_j)^p \br{\frac{1}{ n_i+ n_j} b_i b_j^\dagger }\rho({H}, \tau_1) \br{\frac{1}{ n_i+ n_j} b_i b_j^\dagger }^\dagger } 
&\le 4^{p}   \tr\brr{   \br{4  +  n_i+ n_j }^p \rho({H}, \tau_1)}  \notag \\
&\le 8^p  \tr\brrr{   \brr{4^p  + ( n_i+ n_j)^p} \rho({H}, \tau_1)}  \notag \\
&\le 8^p   \brr{4^p  +2^p C_p }  ,
\label{ineq_bar_pi_bi_bj_2_0_1}
\end{align}
where we use the standing assumption \eqref{eq:rhoconditionst_av} in the last step. By combining the inequalities~\eqref{ineq_bar_pi_bi_bj_2_0} and \eqref{ineq_bar_pi_bi_bj_2_0_1}, we obtain 
\begin{align}
\norm{\bar{\Pi}_{i,j,\bar{q}}^\co b_i b_j^\dagger \sqrt{\rho({H}, \tau_1) } }_F.
\le C\bar{q}^{1-p/2}  ,
\label{ineq_bar_pi_bi_bj_2_0_ffin}
\end{align}
Returning to the inequality~\eqref{interaction/H__0/upper:boud}, we have shown that
\begin{align}
&\abs{\tr \brrr{\brr{ J_{i,j}J_{k,l}\bar{\Pi}_{\tilde{X},\bar{q}} (b_i b_j^\dagger +{\rm h.c.} )\bar{\Pi}_{i,j,\bar{q}}^\co \cdot \bar{\Pi}_{k,l,\bar{q}}^\co (b_{k} b_{l}^\dagger +{\rm h.c.} ) \bar{\Pi}_{\tilde{X},\bar{q}} }\rho(\tilde{H}, \tau_1) }} \le C  \bar{q}^{2-p} .
 \label{interaction/H__0/upper:boud/2}
\end{align}
Finally, by using the assumption on the graph geometry, cf.\ \eqref{supp_parameter_gamma_X_i}, we have
\begin{align}
 \sum_{\substack{i,j,k,l\\  i\in \tilde{X},\ k \in \tilde{X},\ \{i,j\}\cap \{k,l\}\neq \emptyset}} 1 
 &\le \sum_{i\in \tilde{X}}  \sum_{j: \dist_{i,j}=1}  \sum_{k: \dist_{i,k}\le 2} \sum_{l: \dist_{k,l}=1} 1 \le  |\tilde{X}| \gamma^3 2^D ,
\end{align}
which yields
\begin{align}
 \norm{\bar{\Pi}_{\tilde{X},\bar{q}}^\co H_0 \bar{\Pi}_{\tilde{X},\bar{q}} e^{-\mathrm{i}\tilde{H} \tau_1}  \sqrt{\rho}}_F^2
 =\tr \brr{ \bar{\Pi}_{\tilde{X},\bar{q}}  H_0  \bar{\Pi}_{\tilde{X},\bar{q}}^\co H_0 \bar{\Pi}_{\tilde{X},\bar{q}}\rho(\tilde{H}, \tau_1) }\le C |\tilde X| \bar q^{2-p}.
\end{align}
We have thus proved \eqref{main_ineq_ave_norm_H_0_eff_ham1} and hence the claim.
\qed 
%
%
%
%
\section{Conclusions}\label{sect:conclusion}
We considered a class of Bose-Hubbard Hamiltonians with local interactions of the form $n_i^p$ with $p>2$ sufficiently large. We prove that general initial states that are translation-invariant and of uniformly bounded energy density satisfy an enhanced LRB  $v\sim  t^{\frac{D}{p-D-1}}$ compared to \cite{kuwahara2022optimal}  $v\sim t^{D-1}$. In particular, we obtain an almost ballistic LRB $v\sim t^\epsilon$ for sufficiently large $p$. Our bound excludes scenarios for rapid information transport through boson accumulation for our Hamiltonians that were put forward in \cite{kuwahara2022optimal}, cf.\ Figure 4a therein.

While interactions of the form $n_i^p$ with integer $p>2$ are not part of the standard Bose-Hubbard model, they have been discussed in the physics literature to arise physically from local $p$-body repulsion \cite{will2010time, mark2011precision,petrov2014elastic,petrov2014three,mondal2020two}.  Of course, the physically most relevant case is $p=2$, which corresponds to two-body interactions in the standard Bose-Hubbard Hamiltonian, and our result does not cover this case. Addressing Questions 1 and 1' raised in the introduction for $p=2$ will require a much finer dynamical analysis of the local energy flow and, relatedly, the dynamical build-up of higher particle-particle correlations  for suitable initial states. Controlling these other local dynamical quantities will require new insights. 

A different feature of our result is that the LRB on the expectation is appreciably stronger than the one on the trace norm. 
       For long-range quantum spin systems, a similar phenomenon is observed \cite{tran2020hierarchy}.    Of course, the expectation value is bounded by the trace norm, but it could indeed be substantially smaller due to cancellations. While we do not expect such cancellations to occur generically, it is an interesting question under what circumstances the different $t$-scaling truly occurs dynamically, i.e., what kind of lower bounds can be proved. If this phenomenon persists at the level of lower bounds, this would add another unexpected wrinkle to the rich and subtle nature of bosonic information propagation, which would require refining Questions 1 and 1' above based on which measure of size is used.

Finally, looking beyond the specific model, our result places into center focus the significant qualitative rigidity effect that dynamical constraints, specifically symmetries (here, translation-invariance) and energy conservation, have for quantum dynamics of strongly interacting bosons.
Our result is a rigorous manifestation of the associated rigidity concepts in the many-boson context. It would be interesting to explore the effect of different types of dynamical constraints, for instance, dipole moment conservation, which was recently studied in the Bose-Hubbard context in \cite{lake2022dipolar}.

\section*{Acknowledgments}
The authors thank Jingxuan Zhang for useful comments on a draft version of the manuscript.

T.K.\ acknowledges the Hakubi projects of RIKEN.  
T.K.\ is supported by JST PRESTO (Grant No. JPMJPR2116), ERATO (Grant No. JPMJER2302), and JSPS Grants-in-Aid for Scientific
Research (No. JP24H00071), Japan. The research of M.L.\ is partially supported by the DFG through the grant TRR 352 – Project-ID 470903074 and by the European Union (ERC Starting Grant MathQuantProp, Grant Agreement 101163620).\footnote{Views and opinions expressed are however those of the authors only and do not necessarily reflect those of the European Union or the European Research Council Executive Agency. Neither the European Union nor the granting authority can be held responsible for them.}

\section*{Data availability statement}
The manuscript has no associated data.

\section*{Conflict of interest statement}
The authors have no competing interests to declare that are relevant to the content of this article.

\appendix

\section{Examples of translation-invariant quantum states with bounded energy density and significant boson concentration}
\label{app:bad}

We recall that, in proving Theorem \ref{thm:maininformal}, a key input was the polynomial tail bound on the local boson number distribution
\begin{align}
\label{tail_upp_prob}
\tr \brr{\Pi_{ n_i \ge q} \rho_0(t)} \lesssim \frac{1}{q^p}, \qquad \textnormal{for all } t.
\end{align}
which followed via Markov's inequality from the moment bound 
\begin{align}
\label{tr_moment_upp_prob}
\tr \brr{ n_i^p\rho_0(t)}\leq C,\qquad \textnormal{for all } t.
\end{align}
These bounds are derived by using translation-invariance and energy conservation $\tr\brr{\rho H} = \tr\brr{\rho(t) H}$, cf.\  Lemma~\ref{lm:moment}.

Any improvement of the tail bound \eqref{tail_upp_prob} would yield an improved LR velocity. It is therefore a natural question whether our approach is sharp, i.e., whether \eqref{tr_moment_upp_prob} can be improved.
(For short times, Proposition~\ref{prop:interpolation} yields a stronger upper bound, but as noted thereafter, this is insufficient for an improved LR bound.) 

\begin{itemize}
    \item Question 1: Is the current proof via Lemma \ref{lm:moment} sharp, i.e., can we derive a stronger bound than \eqref{tr_moment_upp_prob} assuming translation-invariance and conservation of the first moment of the energy $\tr\brr{\rho H} = \tr\brr{\rho(t) H}$?
    \item Question 2: Can \eqref{tr_moment_upp_prob} be improved by also using conservation of the second moment of the Hamiltonian $\tr\brr{\rho(t) H^2}=\tr\brr{\rho H^2} $?
    \end{itemize}
  
In this Appendix, we construct explicit quantum states which show that the answer to Question 1 is ``Yes'' and the answer to Question 2 is ``No.'' In this sense, the particular method we develop in this paper is sharp; i.e., it will be necessary to leverage other dynamical constraints to obtain further improvements to the bosonic LRB for translation-invariant initial states.

\subsection{Construction of high boson occupation translation-invariant quantum state}
In Proposition \ref{prop:bad_quantum_state} below, we construct an explicit translation-invariant quantum state which saturates the bound~\eqref{tr_moment_upp_prob} while, at the same time, both $\frac{\mathrm{tr}\brr{\rho H}}{|\Lambda_R|}$ and $\frac{\mathrm{tr}\brr{\rho H^2}}{|\Lambda_R|^2}$ are arbitrarily small. 

At the same time, there are many translation-invariant initial states of small particle density for which $\frac{\mathrm{tr}\brr{\rho H}}{|\Lambda_R|}$ and $\frac{\mathrm{tr}\brr{\rho H^2}}{|\Lambda_R|^2}$ are $\mathcal O(1)$ (We give a simple, explicit example after Proposition \ref{prop:bad_quantum_state} for completeness).


Therefore, using only translation-invariance and moment conservation does not allow to exclude that a well-behaved (i.e., one of small local particle density) initial quantum state time-evolves into a badly behaved quantum state (i.e., one that saturates \eqref{tr_moment_upp_prob}). In particular, the construction shows that \textit{translation-invariance by itself} is certainly not enough for controlling the boson concentration better than in \eqref{tr_moment_upp_prob}.\\

For simplicity, we work in two dimensions $D=2$, but the construction generalizes to any dimension in an obvious way.

\begin{prop} \label{prop:bad_quantum_state}
Let $p\geq 2$ and let $\Lambda_R\subset \mathbb Z^2$ be a discrete torus of side length $R$. We choose the Hamiltonian $H$ as a Bose-Hubbard Hamiltonian of the form
\begin{align}
H=J \sum_{i\sim j}(b_i b_j^\dagger +{\rm h.c.}) + U\sum_{i\in \Lambda_R} ( n_i-1)^p 
\end{align}
with periodic boundary conditions.

Then, there exists a constant $C>0$ such that the following holds in the limits where $R,q\to\infty$ with $R\geq q^p$.
There exists a bosonic quantum state $\rho$ with the following properties.
\begin{enumerate}[label=(\Roman*)]
\item{} The state $\rho$ is translation invariant and has total particle number $|\Lambda_R|$ (i.e., the global particle density is $=1$).
\item{} The state $\rho$ satisfies,
\begin{align}
\label{moment_equality_mott}
\frac{\mathrm{tr}\brr{\rho H}}{|\Lambda_R|}=\mathcal O(q^{-p}),\qquad
\frac{\mathrm{tr}\brr{\rho H^2}}{|\Lambda_R|^2}=\mathcal O(q^{-p}),
\qquad \textnormal{for }q\to\infty.
\end{align}
\item{} For any site $i\in\Lambda_R$, we have
 \begin{align}\label{eq:assertionIII}
\tr\brr{ \rho\Pi_{ n_i \ge q}} \ge \frac{C}{q^p} .
\end{align}
\end{enumerate}
\end{prop}

A few remarks are in order.

(i) In assertion (II), the moments of $H$ are normalized in a natural way to obtain finite numbers in the thermodynamic limits $R\to\infty$.

(ii) In the statement, we consider the boson-boson interaction in the form of $(n_i-1)^p$, but this is just for convenience and the construction generalizes to many other forms. In fact, if we consider the Bose-Hubbard type interactions, e.g., $n_i^{p-1}(n_i-1)$, we are able to improve the second order moment bound to $\mathrm{tr}(\rho H^2)-\brr{\mathrm{tr}(\rho H)}^2 =\mathcal O(q^{-p}) |\Lambda_R|$. 

 (iii) The error terms $\mathcal O(q^{-p})$ are mainly due to finite-size \& number-theoretic effects. It is possible to obtain vanishing errors as $R\to\infty$ by a slightly modified, more involved construction.
 

(iv) The proof gives the following stronger version of assertion (III): Let $q=\ell^{1/p}$. Fix any site $i\in \Lambda_R$. For sufficiently large $R$, with probability larger than $1/\ell$, this site belongs to a one-dimensional path of length $R$ with boson number $\sim \ell^{1/p}$ at each of the sites in the path. The fact that the state we construct not only has large occupation on individual sites but on entire one-dimensional paths shows that no improvement can be expected from using the kind of alternative refined LRB for bounded interactions derived in \cite{kuwahara2022optimal} which takes into account the average occupation along a one-dimensional path.\\


For completeness, we give an explicit example of a well-behaved translation-invariant quantum state $\ket{m}$ (i.e., one of small local particle density)  with bounded first two energy moments. 

  We define the vertical lines 
\begin{equation}\label{eq:Ldefn}
    L_x=\left\{(x,y)\in \Lambda_R\,:\, y\in \mathbb Z\right\} .
\end{equation}
In this Appendix, we use Dirac's bra-ket notation. We assume $R$ is even and introduce the following Mott states
\[
\ket{m_{ev}}=\bigotimes_{1\leq x \leq R} \ket{0}_{L_{2x-1}}\otimes \ket{2}_{L_{2x}},\qquad 
\ket{m_{odd}}=\bigotimes_{1\leq x \leq R} \ket{2}_{L_{2x-1}}\otimes \ket{0}_{L_{2x}}.
\]
which both consist of alternating vertical strips of occupation numbers $=0$ and $=2$. Then we consider the state 
\[
\ket{m}=\frac{\ket{m_{ev}}+\ket{m_{odd}}}{\sqrt{2}}.
\]
The state $\ket{m}$ is translation-invariant and has total particle number $|\Lambda_R|$. Moreover, it satisfies the moment bounds\footnote{Here, the second equality holds by
\begin{align}
\bra{m} H^2 
\ket{m}
&= J^2 \sum_{i\sim j} \bra{m} (b_i b_j^\dagger +{\rm h.c.}) ^2\ket{m}  +U^2R^4\notag\\
&= J^2 \sum_{\qmexp{i,j}}  \bra{m}   n_i( n_j+1)+  n_j( n_i+1) \ket{m}+U^2R^4\\  
&=4 J^2 R^2+U^2 R^4 . 
\label{second_moment_H^2}
\end{align}}
\[
\frac{\bra{m}H\ket{m}}{|\Lambda_R|}=U,\qquad \frac{\bra{m}H^2\ket{m}}{|\Lambda_R|^2}=U^2+\frac{4J^2}{|\Lambda_R|}.
\]
By Proposition \ref{prop:bad_quantum_state}, using only translation-invariance and bounds on the first two moments, it is not possible to exclude that the time-evolved state $\ket{m(t)}$ has significantly accumulated bosons in the sense of \eqref{eq:assertionIII}.


\subsection{Proof of Proposition~\ref{prop:bad_quantum_state}.}
We first construct the preliminary state $\ket{\psi_0}$.
Recall \eqref{eq:Ldefn} and define the rectangles $L_{x:x'}$ by  $\bigcup_{x_1\in[x,x']} L_{x_1}$ (see also Fig.~\ref{fig:bad_quantum_state}).  
Fix $\ell\geq 1$.  


 \begin{figure}[tt]
\centering
\includegraphics[clip, scale=0.4]{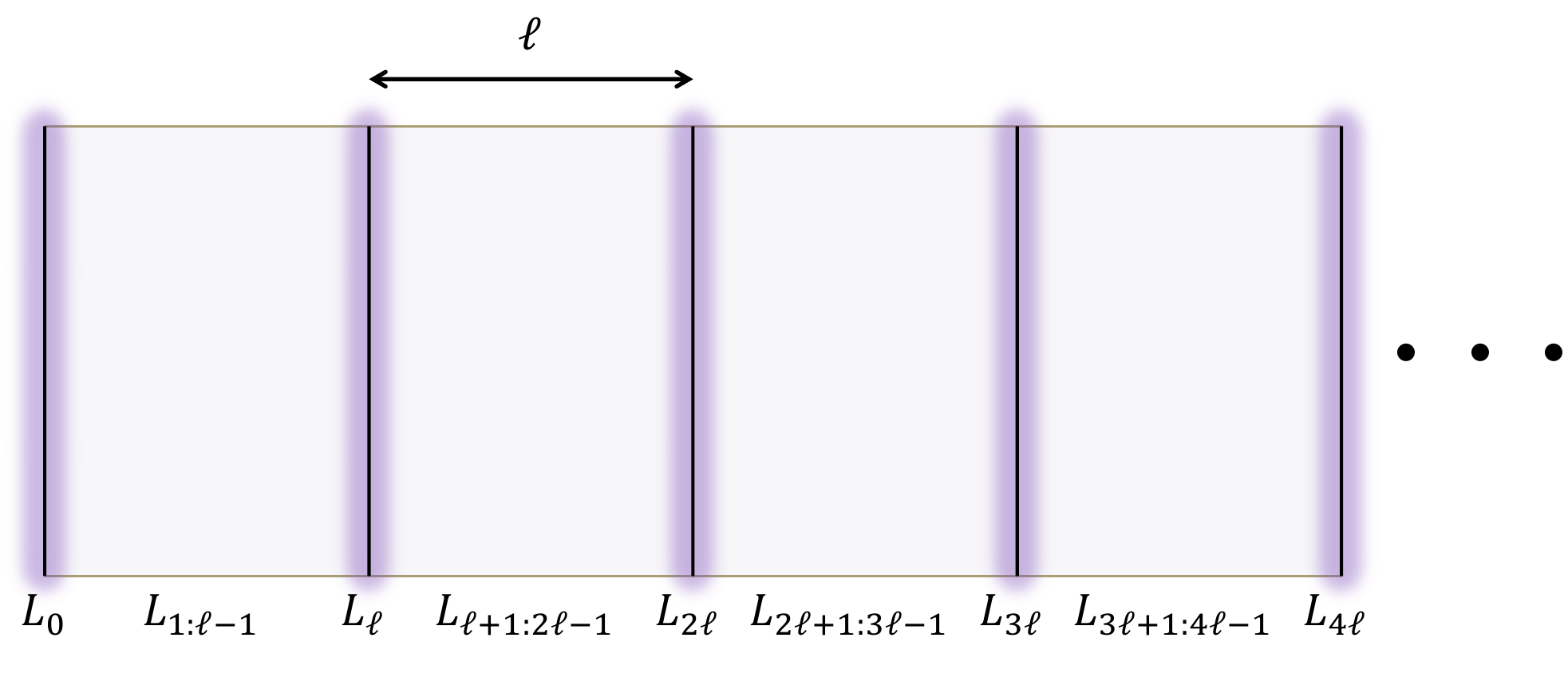}
\caption{Schematic description of the construction of the preliminary quantum state $\ket{\psi_0}$ in Eq.~\eqref{eq:bad_quantum_state_const}. 
In $\ket{\psi_0}$, one-dimensional paths with high-boson density $q$ periodically appear with spacing $\ell$. 
Between these high-density paths, the quantum state is given by a low-energy state with small boson density $\gamma_0<1$.  
The state $\ket{\psi_0}$ is constructed so that the conditions in Eq.~\eqref{moment_equality_mott} are satisfied, but it is not translation invariant. Therefore we average it over translations through Eq.~\eqref{eq:bad_quantum_state_const_trans} so that $\rho$ to obtain a translation-invariant state.
}
\label{fig:bad_quantum_state}
\end{figure}

Let $R\geq \ell\geq 3$. We construct the preliminary quantum state $\ket{\psi_0}$ by alternating high-occupancy lines with the low-energy states from Lemma \ref{lm:phitrialstate} (see Fig.~\ref{fig:bad_quantum_state}). Let $\gamma_0\in(0,1)$ and $q\geq 2$ to be determined later.  For simplicity, we suppose that $R$ is a multiple of $\ell$. We set
 \begin{align}
 \label{eq:bad_quantum_state_const}
\ket{\psi_0}
= \bigotimes_{s=0}^{R/\ell-1} \ket{q_{L_{s\ell} }} \otimes \ket{\phi_{\gamma_0,L_{s\ell+1:(s+1)\ell-1}}} ,
\end{align}
Here $\ket{q_{L_{s\ell} }}$ is the Mott state on the region $L$ with the boson number $q$ at the site of $i\in L_0$. The states on the intermediate rectangles $\ket{\phi_{\gamma_0,L_{s\ell+1:(s+1)\ell-1}}}$ are constructed via the following lemma.

\begin{lemma}[Low-energy states]\label{lm:phitrialstate}
    Let $\gamma_0\in(0,1)$. The following holds as $R\geq \ell\to\infty$. There exist  $e_1,e_2>0$ and a quantum state $\ket{\phi_{\gamma_0,L_{1:\ell-1}}}$ supported on the subset $L_{1:\ell-1}$ with boson density $\gamma_0+\mathcal O(R^{-1})$ such that the following energetic conditions are satisfied:
    \begin{itemize}
\item[(a)] $\ket{\phi_{\gamma_0,L_{1:\ell-1}}}$ commutes with translations in the vertical direction.
    \item[(b)] First-moment condition. 
    \begin{align}
\frac{\bra{\phi_{\gamma_0,L_{1:\ell-1}}} H_{L_{1:\ell-1}}\ket{\phi_{\gamma_0,L_{1:\ell-1}}}}{|L_{1:\ell-1}|}= -\gamma_0 e_1+(1-\gamma_0) e_2  +\mathcal O(\ell^{-1}).
\end{align}
\item[(c)] Second moment condition. 
\begin{align}\label{eq:largevariance}
\frac{\bra{\phi_{\gamma_0,L_{1:\ell-1}}}  H_{L_{1:\ell-1}}^2\ket{\phi_{\gamma_0,L_{1:\ell-1}}}}{|L_{1:\ell-1}|^2} \leq\brr{-\gamma_0 e_1+(1-\gamma_0) e_2}^2+\mathcal O(\ell^{-1}).
\end{align}

\end{itemize}
\end{lemma}
Below, we will choose $\gamma_0\approx 1$ so that the state has negative energy. The proof idea is as follows: First, we construct a trial state of density-$1$ of strictly negative energy. Second, we consider the corresponding ground state and deplete it by adding particle-free regions to obtain a state satisfying the condition (a). 

For now, we assume Lemma \ref{lm:phitrialstate} to hold and continue the proof of Proposition \ref{prop:bad_quantum_state}. We link the two parameters $q$ and $\gamma_0$ in the state $\ket{\psi_0}$ by the condition
that the total boson density equals $1$, i.e., 
\[
q |L_0| +\brr{\gamma_0+\mathcal O(R^{-1})} |L_{1:\ell-1}| =  |L_{0:\ell-1}|,
\]
which we can rephrase as
\begin{align}\label{eq:constraint1}
q=(1 - \gamma_0 )(\ell-1)+ 1+\mathcal O(R^{-1}).
\end{align}


We construct from $\ket{\psi_0}$ a  translation invariant state by averaging it over translations up to order $\ell-1$,
 \begin{align}
  \label{eq:bad_quantum_state_const_trans}
  \rho:=\frac{1}{\ell} \sum_{x=0}^{\ell-1} \ket{\psi_x}\bra{\psi_x} ,\qquad  \ket{\psi_x}:=\bigotimes_{s=0}^{R/\ell-1}  \ket{q_{L_{x+s\ell} }} \otimes \ket{\phi_{\gamma_0,x+L_{s\ell+1:x+(s+1)\ell-1}}}.
\end{align}
By construction, the state is periodic in the horizontal direction, and by property (a) of Lemma \ref{lm:phitrialstate}, it is periodic in the vertical direction. Therefore, it satisfies
assertion (I) in Proposition \ref{prop:bad_quantum_state}. 

In the quantum state $\ket{\psi}$, the reduced density matrix on the  line $L_0$ is given by
 \begin{align}
\rho_{L_0}=\frac{1}{\ell} \ket{q_{L_0}}\bra{q_{L_0}} +\frac{\ell-1}{\ell}  \tilde{\rho}_{L_0},
\end{align}
where $\tilde{\rho}_{L_0} $ is constructed from the reduced density matrix of $\ket{\phi_{\gamma_0,x+L_{s\ell+1:x+(s+1)\ell-1}}}$.
In particular, we see that assertion (III) holds\footnote{In fact, recalling the definition of $\ket{\psi}$, we see that $q\gtrsim \ell^{1/p}$ implies the stronger version of assertion (III) that the site belongs to an entire path of boson numbers $\gtrsim \ell^{1/p}$.} if we ensure that
$
q\geq C\ell^{1/p}.
$
Using \eqref{eq:constraint1}, this can be rephrased as the constraint
\begin{equation}\label{eq:constraint1'}
(1 - \gamma_0 ) (\ell-1)+ 1+\mathcal O(R^{-1})\geq C\ell^{1/p}.
\end{equation}
For large $\ell$ and $\gamma_0\approx 1$, this constraint is asymptotic to $\gamma_0\gtrsim 1-C\ell^{\frac{1-p}{p}}$.\\

Hence, to prove Proposition \ref{prop:bad_quantum_state}, it remains to find a choice of $\ell$ so that simultaneously \eqref{eq:constraint1'} holds and the energetic relations in~\eqref{moment_equality_mott} from assertion (II) are satisfied.

\underline{First moment.}
By translation-invariance, the average energy is given by
 \begin{align}
\tr\brr{\rho H} = \bra{\psi_0} H \ket{\psi_0} .
\end{align}
We decompose $H$ as 
 \begin{align}
H = \sum_{s=0}^{R/\ell-1}  \br{ H_{L_{s\ell+1:(s+1)\ell-1}} + \partial h_{L_{s\ell} }+ V_{L_{s\ell} }},
\end{align}
where $ \partial h_{L_{s\ell} }$ are the hopping terms that act on the line $L_{s\ell}$. Note that $\bra{\psi_0}\partial h_{L_{s\ell} } \ket{\psi_0}=0$.
We have
 \begin{align}
\bra{\psi_0}H \ket{\psi_0} &=  \frac{R}{\ell} \brr{ (-\gamma_0e_1+(1-\gamma_0) e_2+\mathcal O(\ell^{-1})} |L_{1:\ell-1}|+ (q-1)^p |L_{0} |    .
\end{align}
We introduce the following constraint on the parameter  $\gamma_0$
\[
\brr{-\gamma_0e_1+(1-\gamma_0) e_2}  |L_{1:\ell-1} | + (q-1)^p |L_0 |=0  
\]
Substituting for $q$ by using the previous constraint \eqref{eq:constraint1}, we rephrase this as
\begin{align}\label{eq:constraint2}
\gamma_0 e_1\ell = (1-\gamma_0)^p(\ell-1)^p +\brr{(1-\gamma_0) e_2}\ell +\mathcal O(\ell^{-1}).
\end{align}
For large $\ell$ and $\gamma_0\approx 1$, this amounts to $e_1\ell \approx (1-\gamma_0)^p \ell^p$ or equivalently $\gamma_0\approx 1-e_1^{1/p} \ell^{\frac{1-p}{p}}$. Absorbing the case of small $\ell$ in the constant, we can ensure that \eqref{eq:constraint1'} and \eqref{eq:constraint2} hold simultaneously. This establishes
\[
\tr\brr{\rho H}  =\mathcal O(\tfrac{R^2}{\ell}).
\]

\underline{Second moment.} 
By translation-invariance, 
 \begin{align}
\tr\brr{\rho H^2} = \bra{\psi_0} H^2 \ket{\psi_0} .
\end{align}
For the state $\ket{\psi_0}$, we have
 \begin{align}
H \ket{\psi_0} =  \sum_{s=0}^{R/\ell-1} \br{ H_{L_{s\ell+1:(s+1)\ell-1}}  + q^p |L_{s\ell} | } \ket{\psi_0} 
+  \partial h_{L_{s\ell}} \ket{\psi_0} .
\end{align}
Using this and $\bra{\psi_0}  H_{L_{s\ell+1:(s+1)\ell-1}} \partial h_{L_{s\ell}}  \ket{\psi_0} =0$, we find
 \begin{align}
&\bra{\psi_0}  H^2 \ket{\psi_0} \\ 
=&\bra{\psi_0}  H^2 \ket{\psi_0} -\bra{\psi_0}  H \ket{\psi_0}^2+\mathcal O(\tfrac{R^4}{\ell^2})\\
=&\sum_{s=0}^{R/\ell-1}\Big(\bra{\phi_{\gamma_0,L_{s\ell+1:(s+1)\ell-1}}} H_{L_{s\ell+1:(s+1)\ell-1}}^2\ket{\phi_{\gamma_0,L_{s\ell+1:(s+1)\ell-1}}}\\
&\qquad +2q^p |L_{s\ell} |  \cdot (-\gamma_0e_1 +(1-\gamma_0)e_2 +\mathcal O(\ell^{-1})) + (q^p |L_{s\ell}|)^2
+ \bra{\psi_0}\partial h_{L_{s\ell}}^2 \ket{\psi_0} \Big)+\mathcal O(\tfrac{R^4}{\ell^2})\notag \\
=&\frac{R}{\ell} \br{  \bra{\phi_{\gamma_0,L_{1:\ell-1}}} \Delta H_{L_{1:\ell-1}}^2 \ket{\phi_{\gamma_0,L_{1:\ell-1}}} 
+ \bra{\psi_0}\partial h_{L_{\ell}}^2 \ket{\psi_0}}+\mathcal O(\tfrac{R^4}{\ell^2}),
\label{eq:DeltaHintro}
\end{align}
where we introduced
\[
\Delta H_{L_{1:\ell-1}}^2=H_{L_{1:\ell-1}}^2- 2q^p |L_{0} |  \cdot (-\gamma_0e_1 +(1-\gamma_0)e_2)|L_{1:\ell-1}| +(q^p |L_{0}|)^2.
\]
We estimate $ \bra{\phi_{\gamma_0,L_{1:\ell-1}}} \Delta H_{L_{1:\ell-1}}^2 \ket{\phi_{\gamma_0,L_{1:\ell-1}}}$ via \eqref{eq:largevariance} in part (b) of Lemma \ref{lm:phitrialstate}. By our choice \eqref{eq:constraint1} of $q$ (i.e., $q\propto \ell^{1/p}$), the order $|L_{1:\ell-1}|^2$ terms cancel and so
\[
\bra{\phi_{\gamma_0,L_{1:\ell-1}}} \Delta H_{L_{1:\ell-1}}^2\ket{\phi_{\gamma_0,L_{1:\ell-1}}} =\mathcal O(R^2\ell).
\]
By Cauchy-Schwarz,
 \begin{align}
\bra{\psi_0}\partial h_{L_0}^2 \ket{\psi_0}
&\leq J^2 \sum_{ \{i,j\}\cap L_0 \neq \emptyset}  \bra{\psi_0}   n_i( n_j+1)+  n_j( n_i+1) \ket{\psi_0} \notag \\
\label{eq:boundaryhoppingCS}
& \le 3J^2  q(q+1) |L_0| \leq CJ^2  \ell^{2/p} R.
\end{align}
 It follows that $
\bra{\psi_0}  H^2 \ket{\psi_0}\leq \mathcal O(R^3)+\mathcal O(\tfrac{R^4}{\ell^2})$.
This completes the proof of Proposition \ref{prop:bad_quantum_state}. 
\qed

\subsection{Proof of Lemma \ref{lm:phitrialstate}}
The proof is comprised of two steps. First, we construct an explicit trial state of particle density $1$ and negative energy expectation. Then we consider the density-$1$ ground state and lower its density to $\gamma_0$ by ``dilution'', i.e.,  tensoring it with zero-occupancy regions in a suitable way. For technical reasons (to control certain boundary terms via the number of particles on the boundary), we need translation-invariance of these subsystem states and for this reason, we use a periodic subsystem Hamiltonian, which again introduces other boundary terms; see Step 2 for the details. 

\underline{Step 1.} Let $\ell_0\geq 4$.
We consider the following trial state $\ket{\varphi}$ of particle density $1$:
\[
\ket{\varphi} :=\bigotimes_{k=1}^R\left( \bigotimes_{j=1}^{\lfloor\ell_0/2\rfloor}\ket{\varphi_{(2j-1,k),(2j,k)}} \otimes \ket{1}^{\ell_0/2-\lfloor \ell_0/2\rfloor}\right),
\]
where $k$ denotes the vertical coordinate and $\ket{\varphi_{(2j-1,k),(2j,k)}} $ is given in the particle occupation basis by
\[
\ket{\varphi_{(2j-1,k),(2j,k)}} =\lambda_1 \ket{1,1} - \lambda_2 {\rm sign}(J)  (\ket{2,0}+\ket{0,2}),
\]
with $\lambda_1,\lambda_2\ge 0$ and $\lambda_1^2+2\lambda_2^2=1$. 

We consider the periodized Hamiltonian
\[
H^{\mathrm{per}}_{L_{1:\ell_0-1}}:=H_{L_{1:\ell_0-1}}+\partial h_{\ell_0\leftrightarrow 1}.
\]
We calculate $\bra{\varphi} H_{L_{1:\ell_0-1}} \ket{\varphi} 
$ and show that it is strictly negative and on the order of the volume $|L_{1:\ell_0-1}|$ for sufficiently large $\ell_0$.

For the potential energy, we have
\[
\bra{\varphi} U\sum_{x\in L_{1:\ell-1}}  (n_x-1)^2 \ket{\varphi} = 2 U R\ell \lambda_2^2.
\]
For the kinetic energy, we distinguish cases. Denote $h_{(j,k),(j',k')}=J(b_{(j,k)}^\dagger b_{(j',k')}+ {\rm h.c.})$. Straightforward calculations in the bulk ($1\leq j<\ell/2-1$) gives
\[\begin{aligned}
\bra{\varphi} h_{(2j,k),(2j+1,k)} \ket{\varphi} =\bra{\varphi} h_{(2j,k),(2j,k+1)} \ket{\varphi}   =0.
\end{aligned}
\]
That is, the only non-vanishing hopping terms are those ``within'' a single tensor factor $\ket{\varphi_{(2j-1,k),(2j,k)}}$. In the bulk, i.e., for $1\leq j<\ell/2-1$, these are given by
\[
\bra{\varphi} h_{(2j-1,k),(2j,k)} \ket{\varphi} 
=-4 |J| \lambda_1 \lambda_2 . 
\]
We estimate the hopping terms at the boundary by Cauchy-Schwarz and obtain
\[
\begin{aligned}
\bra{\varphi} H^{\mathrm{per}}_{L_{1:\ell_0-1}} \ket{\varphi} 
&\leq R\ell_0 ( 2 U \lambda_2^2 - 2|J|\lambda_1 \lambda_2)+c(J,U) R\\
 &=2UR\ell_0 (\lambda_2^2 - T_1 \lambda_1\lambda_2)   +c(J,U) R,\qquad \textnormal{with }\ T_1=\frac{|J|}{U}.
 \end{aligned}
\]
Minimizing over pairs $(\lambda_1,\lambda_2)$ satisfying the constraint $\lambda_1^2+2\lambda_2^2=1$, we obtain
\[
\min_{\lambda_1,\lambda_2}\bra{\varphi} H^{\mathrm{per}} _{L_{1:\ell_0-1}} \ket{\varphi} = \frac{U R\ell_0}{2} \Bigl(1- \sqrt{1+2T_1^2}  \Bigr)+c(J,U) R.
\]
Hence, for all sufficiently large $\ell_0$, there exists $\tilde e_1>0$ such that
\[
\inf\mathrm{spec}\,  H^{\mathrm{per}}_{L_{1:\ell_0-1}} \leq -\tilde e_1 |L_{1:\ell_0-1}|.
\]
The Hamiltonian $H^{\mathrm{per}}_{L_{1:\ell_0-1}}$ commutes with $n_{L_{1:\ell_0-1}}$, the particle number on $L_{1:\ell_0-1}$ and so we can block-diagonalize it. Since the trial state $\ket{\varphi}$ satisfies $n_{L_{1:\ell_0-1}}\ket{\varphi}=|L_{1:\ell_0-1}|\ket{\varphi}$ (it has global particle density $1$), we may conclude
\begin{equation}\label{eq:step1}
\inf\mathrm{spec}\, (  \Pi_{n_{L_{1:\ell_0-1}}=|L_{1:\ell_0-1}|} H^{\mathrm{per}}_{L_{1:\ell_0-1}} \Pi_{n_{L_{1:\ell_0-1}}=|L_{1:\ell_0-1}|})\leq -\tilde e_1 |L_{1:\ell_0-1}|.
\end{equation}
where $ \Pi_{n_{L_{1:\ell_0-1}}=|L_{1:\ell_0-1}|}$ projects onto the spectral subspace on which $n_{L_{1:\ell_0-1}}=|L_{1:\ell_0-1}|$.

\underline{Step 2.} 
Fix $\gamma_0\in (0,1)$. We set $\ell_0=\lfloor\gamma_0 \ell \rfloor<\ell$.  We assume that $\ell$ is large enough so that \eqref{eq:step1} holds. Using  \eqref{eq:step1}, the spectral theorem implies that exists a normalized ground state $\ket{\xi_{1:\ell_0-1}}$ of global particle density $1$ and energy per volume $- e_1\leq -\tilde e_1<0$, i.e.,
\[
H^{\mathrm{per}}_{L_{1:\ell_0-1}}\ket{\xi_{1:\ell_0-1}}=-e_1 |L_{1:\ell_0-1}|\ket{\xi_{1:\ell_0-1}}.
\]
By block-diagonalizing $H^{\mathrm{per}}_{L_{1:\ell_0-1}}$ with respect to translations in the vertical direction, we may assume that $\ket{\xi_{1:\ell_0-1}}$ is invariant under translations in the vertical direction up to a phase.

To construct $\ket{\phi_{\gamma_0,0:\ell-1}}$, we dilute $\ket{\xi_{0:\ell_0-1}}$ with particle-free regions as follows.
\[
\ket{\phi_{\gamma_0,0:\ell-1}}:=\ket{\xi_{0:\ell_0-1}}\otimes\ket{0}_{\ell_0:\ell-1},
\qquad \textnormal{with }\ket{0}_{\ell_0:\ell-1}:=\bigotimes_{\ell_0\leq j\leq \ell-1}\bigotimes_{1\leq k\leq R} \ket{0_{j,k}}.
\] 

First, we note that the state $\ket{\phi_{1:\ell}}$ is invariant under translations in the vertical direction up to a phase, which ensures property (a) of Lemma \ref{lm:phitrialstate}.
Moreover, it has particle density $\frac{\ell_0}{\ell}=\gamma_0+\mathcal O(\ell^{-1})$ as desired.

We check property (b). We abbreviate $H_{a:b}\equiv H_{L_{a:b}}$ and $H^{\mathrm{per}}_{a:b}\equiv H^{\mathrm{per}}_{L_{a:b}}$. The building blocks of $\ket{\phi_{1:\ell}}$ are all individually eigenstates of suitable local Hamiltonians, 
\begin{equation}\label{eq:evalueapp}
\begin{aligned}
H^{\mathrm{per}}_{1:\ell_0-1}\ket{\xi_{{1:\ell_0-1}}}=&- e_1 \ell_0 R\ket{\xi_{{1:\ell_0-1}}},\\
 H_{\ell_0:\ell-2}\ket{0}_{\ell_0:\ell-1}=&U\ket{0}_{\ell_0:\ell-1}.\qquad\end{aligned}
\end{equation}
We write $H_{1:\ell-1}=H^{\mathrm{per}}_{1:\ell_0-1}+H_{\ell_0:\ell-1}+\partial h_{\ell_0-1\leftrightarrow \ell_0}-\partial h_{\ell_0-1\leftrightarrow 1}$ and use \eqref{eq:evalueapp}  to calculate
\begin{equation}\label{eq:firstmomentabove}
\begin{aligned}
&\frac{\bra{\phi_{\gamma_0,L_{1:\ell-1}}} H_{{1:\ell-1}}\ket{\phi_{\gamma_0,L_{1:\ell-1}}}}{|L_{1:\ell-1}|}\\
= &-\frac{\ell_0 R}{{|L_{1:\ell-1}|}} e_1+\frac{|L_{1:\ell-1}|-\ell_0 R}{{|L_{1:\ell-1}|}}  U+\frac{\bra{\phi_{\gamma_0,L_{1:\ell-1}}}  \partial h_{\ell_0-1\leftrightarrow \ell_0}-\partial h_{\ell_0-1\leftrightarrow 1}\ket{\phi_{\gamma_0,L_{1:\ell-1}}}}{|L_{1:\ell-1}|}\\
=& -\gamma_0 e_1+(1-\gamma_0) U  +\mathcal O(\ell^{-1}) +\frac{\bra{\phi_{\gamma_0,L_{1:\ell-1}}}  \partial h_{\ell_0-1\leftrightarrow 1}\ket{\phi_{\gamma_0,L_{1:\ell-1}}}}{|L_{1:\ell-1}|}.
\end{aligned}
\end{equation}
In the last step, we used  $\frac{\ell_0}{\ell}=\gamma_0+\mathcal O(\ell^{-1})$ and the fact that $\bra{\phi_{\gamma_0,L_{1:\ell-1}}} \partial h_{\ell_0-1\leftrightarrow \ell_0}\ket{\phi_{\gamma_0,L_{1:\ell-1}}}=0$. By Cauchy-Schwarz, 
\begin{equation}\label{eq:boundaryest}
\begin{aligned}
\bra{\phi_{\gamma_0,L_{1:\ell-1}}}  \partial h_{\ell_0-1\leftrightarrow 1}\ket{\phi_{\gamma_0,L_{1:\ell-1}}}
\leq& CJ \bra{\phi_{\gamma_0,L_{1:\ell-1}}} (n_{L_{\ell_0-1}}+n_{L_1})\ket{\phi_{\gamma_0,L_{1:\ell-1}}}\\
=&\frac{2CJ}{\ell_0-1}\bra{\xi_{1:\ell_0-1}} n_{L_{1:\ell_0-1}}\ket{\xi_{1:\ell_0-1}}\\
=& 2CJ R
\end{aligned}
\end{equation}
for an explicit constant $C>0$ that may change from line to line that does not depend on any parameters, where we use $\bra{\xi_{1:\ell_0-1}} n_{L_{1:\ell_0-1}}\ket{\xi_{1:\ell_0-1}}=R|L_{1:\ell_0-1}|$. The first equality in \eqref{eq:boundaryest} holds by the periodicity of $\ket{\xi_{1:\ell_0-1}}$ in the horizontal direction. Returning to \eqref{eq:firstmomentabove} and using that $\frac{R}{|L_{1:\ell-1}|}=\mathcal O(\ell^{-1})$, this proves that property (b) holds with $e_2:=U$.

We come to property (c). Using the eigenvalue equations \eqref{eq:evalueapp}, the only non-trivial contributions again come from the hopping term along the boundaries:
\[\begin{aligned}
&\bra{\phi_{\gamma_0,L_{1:\ell-1}}}  H_{L_{1:\ell-1}}^2\ket{\phi_{\gamma_0,L_{1:\ell-1}}}\\
=&\bra{\phi_{\gamma_0,L_{1:\ell-1}}}  (H^{\mathrm{per}}_{1:\ell_0-1}+H_{\ell_0:\ell-1}+\partial h_{\ell_0-1\leftrightarrow \ell_0}-\partial h_{\ell_0-1\leftrightarrow 1})^2\ket{\phi_{\gamma_0,L_{1:\ell-1}}}\\
=&\bra{\phi_{\gamma_0,L_{1:\ell-1}}}  (H^{\mathrm{per}}_{1:\ell_0-1}+H_{\ell_0:\ell-1})^2\ket{\phi_{\gamma_0,L_{1:\ell-1}}}\\
&+\bra{\phi_{\gamma_0,L_{1:\ell-1}}}  \{H^{\mathrm{per}}_{1:\ell_0-1}+H_{\ell_0:\ell-1},\partial h_{\ell_0-1\leftrightarrow \ell_0}-\partial h_{\ell_0-1\leftrightarrow 1}\}\ket{\phi_{\gamma_0,L_{1:\ell-1}}}\\
&+\bra{\phi_{\gamma_0,L_{1:\ell-1}}} (\partial h_{\ell_0-1\leftrightarrow \ell_0}-\partial h_{\ell_0-1\leftrightarrow 1})^2\ket{\phi_{\gamma_0,L_{1:\ell-1}}}\\
= &(-\gamma_0 e_1+(1-\gamma_0) e_2+\mathcal O(\ell^{-1}))^2|L_{0:\ell-1}|^2\\
 &+2(-\gamma_0 e_1+(1-\gamma_0) e_2+\mathcal O(\ell^{-1}))|L_{0:\ell-1}| \bra{\phi_{\gamma_0,L_{1:\ell-1}}} (\partial h_{\ell_0-1\leftrightarrow \ell_0}-\partial h_{\ell_0-1\leftrightarrow 1}) \ket{\phi_{\gamma_0,L_{1:\ell-1}}}\\
&+\bra{\phi_{\gamma_0,L_{1:\ell-1}}} (\partial h_{\ell_0-1\leftrightarrow \ell_0}-\partial h_{\ell_0-1\leftrightarrow 1})^2 \ket{\phi_{\gamma_0,L_{1:\ell-1}}},
\\
\end{aligned}
\]
where we use Eq.~\eqref{eq:evalueapp} in the last equation.
The boundary terms can be controlled by an extension of the argument in \eqref{eq:boundaryest}.
For the first-order boundary terms, this is immediate,
\begin{align}
 |\bra{\phi_{\gamma_0,L_{1:\ell-1}}} (\partial h_{\ell_0-1\leftrightarrow \ell_0}-\partial h_{\ell_0-1\leftrightarrow 1}) \ket{\phi_{\gamma_0,L_{1:\ell-1}}}|
 &\leq CJ |\bra{\phi_{\gamma_0,L_{1:\ell-1}}} (n_{L_1}+n_{L_{\ell_0-1}})\ket{\phi_{\gamma_0,L_{1:\ell-1}}}| \notag \\
 &\leq 2CJR. \notag 
\end{align}
Together with the factor of $|L_{0:\ell-1}|$ this gives a subleading term $\mathcal O(\ell^{-1})$ compared to the leading $|L_{0:\ell-1}|^2$ term.

For the squared boundary terms $\bra{\phi_{\gamma_0,L_{1:\ell-1}}}  (\partial h_{\ell_0-1\leftrightarrow \ell_0}-\partial h_{\ell_0-1\leftrightarrow 1})^2\ket{\phi_{\gamma_0,L_{1:\ell-1}}}$, we need to use an energy bound. First, multiple applications of Cauchy-Schwarz give 
\[
\bra{\phi_{\gamma_0,L_{1:\ell-1}}} (\partial h_{\ell_0-1\leftrightarrow \ell_0}-\partial h_{\ell_0-1\leftrightarrow 1})^2 \ket{\phi_{\gamma_0,L_{1:\ell-1}}}\leq 
32CR J^2 \sum_{x\in L_1\cup L_{\ell_0-1}}\bra{\phi_{\gamma_0,L_{1:\ell-1}}} n_x^2\ket{\phi_{\gamma_0,L_{1:\ell-1}}},
\]
 with $C>0$ a universal constant. For instance, denoting $\tilde n_x=n_x+1$,
\[
\begin{aligned}
&\bra{\phi_{\gamma_0,L_{1:\ell-1}}} \partial h_{\ell_0-1\leftrightarrow 1}^2 \ket{\phi_{\gamma_0,L_{1:\ell-1}}}\\
=&J^2\sum_{1\leq k,k'\leq R}\bra{\phi_{\gamma_0,L_{1:\ell-1}}} 
( b_{1,k}^\dagger b_{\ell_0-1,k}+b_{1,k} b_{\ell_0-1,k}^\dagger)
( b_{1,k'}^\dagger b_{\ell_0-1,k'}+b_{1,k'} b_{\ell_0-1,k'}^\dagger)
 \ket{\phi_{\gamma_0,L_{1:\ell-1}}}\\
\leq &CJ^2\sum_{1\leq k<k'\leq R}\bra{\phi_{\gamma_0,L_{1:\ell-1}}} (
\tilde n_{1,k}\tilde n_{1,k'}+\tilde n_{1,k}\tilde n_{\ell_0-1,k'}+\tilde n_{\ell_0-1,k}\tilde n_{1,k'}+\tilde n_{\ell_0-1,k}\tilde n_{\ell_0-1,k'}) \ket{\phi_{\gamma_0,L_{1:\ell-1}}}\\
\leq &2CR J^2\sum_{x\in L_1\cup L_{\ell_0-1}}\bra{\phi_{\gamma_0,L_{1:\ell-1}}} (n_x+1)^2 \ket{\phi_{\gamma_0,L_{1:\ell-1}}}.
\\
\leq &8CR J^2\sum_{x\in L_1\cup L_{\ell_0-1}}\bra{\phi_{\gamma_0,L_{1:\ell-1}}} n_x^2 \ket{\phi_{\gamma_0,L_{1:\ell-1}}}.
\\
\end{aligned}
\]
The other boundary terms are estimated analogously.
To control $R\sum_{x\in L_j}\bra{\phi_{\gamma_0,L_{1:\ell-1}}} n_x^2\ket{\phi_{\gamma_0,L_{1:\ell-1}}}$, we first recall periodicity of $\ket{\xi_{1:\ell_0-1}}$ in the horizontal direction, 
\[
R\sum_{x\in L_j}\bra{\phi_{\gamma_0,L_{1:\ell-1}}} n_x^2\ket{\phi_{\gamma_0,L_{1:\ell-1}}}
=\frac{R}{\ell_0}
 \bra{\xi_{1:\ell_0-1}}\sum_{x\in L_{1:\ell_0-1}}n_x^2\ket{\xi_{1:\ell_0-1}}. 
\]
We control $\sum_{x\in L_{0:\ell_0-1}}n_x^2$ by the energy of $\ket{\xi_{1:\ell_0-1}}$ up to lower order terms that are linear in particle number. Indeed, by Cauchy-Schwarz, we have the operator inequality
\[
\sum_{x\in L_{1:\ell_0-1}}n_x^2\leq \frac{1}{U}\left( H^{\mathrm{per}}_{1:\ell_0-1}+c_1(J,U) n_{L_{1:\ell_0-1}}+c_2(J,U)|L_{1:\ell_0-1}|\right)
\]
for suitable constants $c_1(J,U),c_2(J,U)>0$. This implies
\[
\frac{R}{\ell_0}
 \bra{\xi_{1:\ell_0-1}}\sum_{x\in L_{1:\ell_0-1}}n_x^2\ket{\xi_{1:\ell_0-1}}
 \leq\frac{R|L_{1:\ell_0-1}|}{\ell_0 U}(c_1(J,U)+c_2(J,U)-e_1)
 =:c(J,U) R^2.
\]
Since $R^2=|L_{1:\ell_0-1}|^2\mathcal O(\ell_0^{-2})=|L_{1:\ell_0-1}|^2\mathcal O(\ell^{-2})$ from $\ell_0=\lfloor \gamma_0 \ell\rfloor \propto \ell$, this establishes property (c) and completes the proof of Lemma \ref{lm:phitrialstate}. \qed

\begin{footnotesize}
\bibliographystyle{alpha}
\bibliography{LR_boson.bib} 
\end{footnotesize}
%
%
%





%
%
%
%
%
%

\end{document}